\documentclass{article}

\usepackage{amsmath}
\usepackage{amssymb}
\usepackage{array}
\usepackage{pre}
\usepackage[colorinlistoftodos,prependcaption,textsize=tiny]{todonotes}

\title{Parity $\notin \qaczm$ $\iff$ $\qaczm$ is Fourier-Concentrated}

\def\ANON{0} 
\def\COMM{0} 

\ifnum\ANON=1
\author{
  Anonymous Authors.
}
\else
\author{
  Lucas Gretta \thanks{University of California at Berkeley. \ Email: \url{lucas_gretta@berkeley.edu}. \ Supported by NSF Award CCF-2231095.}
  \and
  Meghal Gupta\thanks{University of California at Berkeley. \ Email: \url{meghal@berkeley.edu}. \ Supported by NSF GRFP.}
  \and
  Malvika Raj Joshi\thanks{University of California at Berkeley. \ Email: \url{malvika@berkeley.edu}. \ Supported by NSF grant 2311733 and DOE grant DE-SC0024124.}
}
\fi
\date{}

\ifnum\COMM=1
\newcommand{\malvika}[1]{\textcolor{purple}{[\textbf{Malvika:} {#1}]}}
\newcommand{\luke}[1]{\textcolor{blue}{[\textbf{Luke:} {#1}]}}
\newcommand{\meghal}[1]{\textcolor{red}{[\textbf{Meghal:} {#1}]}}
\newcommand{\gpt}[1]{\textcolor{orange}{[\textbf{GPT:} {#1}]}}
\newcommand{\todomal}[1]{\todo[linecolor=Plum,backgroundcolor=Plum!25,bordercolor=Plum]{\textbf{@mal todo}: #1}}
\newcommand{\todoluke}[1]{\todo[linecolor=blue,backgroundcolor=blue!25,bordercolor=blue]{\textbf{@luke todo}:#1}}
\newcommand{\todomeghal}[1]{\todo[linecolor=red,backgroundcolor=red!25,bordercolor=red]{\textbf{@meghal todo} #1}}
\newcommand{\todoany}[1]{\todo[linecolor=orange,backgroundcolor=orange!25,bordercolor=orange]{\textbf{todo @plz do}: #1}}
\else
\newcommand{\malvika}[1]{}
\newcommand{\luke}[1]{}
\newcommand{\meghal}[1]{}
\newcommand{\gpt}[1]{}
\newcommand{\todomal}[1]{}
\newcommand{\todoluke}[1]{}
\newcommand{\todomeghal}[1]{}
\newcommand{\todoany}[1]{}
\fi

\begin{document}

\maketitle
\ifnum\ANON=1
\PackageWarningNoLine{Global}{Note anonymous mode}
\ifnum\COMM=1
\textcolor{red}{\textbf{Warning:}{ Comments are enabled in anonymous mode}}
\fi
\fi
\ifnum\COMM=1
\PackageWarningNoLine{Global}{Note comments are enabled}
\fi

\begin{abstract}
A major open problem in understanding shallow quantum circuits ($\QACZ$) is whether they can compute $\PARITY$. In this work, we show that this question is solely about the Fourier spectrum of $\qaczm$: \emph{any} $\qaczm$ circuit with non-negligible high-level Fourier mass suffices to exactly compute $\PARITY$ in $\QACZ$. Thus, proving a quantum analog of the seminal LMN theorem for $\ACZ$ \cite{lmn1993ac0} is \emph{necessary} to bound the quantum circuit complexity of $\PARITY$. 

In the other direction, LMN does not fully capture the limitations of $\ACZ$. For example, despite $\MAJORITY$ having $99\%$ of its weight on low-degree Fourier coefficients, no $\ACZ$ circuit can non-trivially correlate with it \cite{razborov1987lower, smolensky1987algebraic}. 
In contrast, we provide a $\QACZ$ circuit that achieves $(1-o(1))$ correlation with $\MAJORITY$, establishing the first average-case decision separation between $\ACZ$ and $\QACZ$. This suggests a uniquely quantum phenomenon: unlike in the classical setting, Fourier concentration may largely characterize the power of $\QACZ$.

$\PARITY$ is also known to be equivalent in $\QACZ$ to inherently quantum tasks such as preparing GHZ states to high fidelity. We extend this equivalence to a broad class of state-synthesis tasks. We demonstrate that existing measures such as trace distance, fidelity, and mutual information are insufficient to capture these states and introduce a new measure, \emph{felinity}. We prove that preparing any state with non-negligible \emph{felinity}, or derived states such as  $\poly(n)$-weight Dicke states, implies $\PARITY \in \qaczm$.

\end{abstract}
\tableofcontents
\section{Introduction}

Understanding the computational power of constant-depth quantum circuits ($\QACZ$) remains a central open problem in quantum complexity theory. 

Classical constant-depth circuits ($\ACZ$) cannot compute $\PARITY$ \cite{Ajtai83,FSS84,hastad1986switch} and in fact, form a strict hierarchy. Adding $\PARITY_n$ gates to $\ACZ$ does not enable other hard computations such as arbitrary modular arithmetic or $\MAJORITY_n$, making $\ACZ \subsetneq \ACZ[2] \subsetneq \ACC \subseteq \TCZ$ \cite{razborov1987lower, smolensky1987algebraic}.  In contrast, $\QACZ$ exhibits an “all-or-nothing” behavior \cite{takahashi2012collapse}.  The class  $\QACZF$, obtained by augmenting $\QACZ$ with $\PARITY$ gates (or its Hadamard conjugate $\FANOUT$), can simulate all of $\TCZ$ \cite{moore1999qac0, hoyer2005fanout}. Thus, determining whether $\QACZ$ can compute $\PARITY$, or equivalently, whether $\QACZ = \QACZF$, is a crucial step.

A distinctive feature of $\PARITY$ is that all of its Fourier mass lies on the highest-degree coefficient. However, $\ACZ$ circuits exhibit strong concentration on low-degree Fourier coefficients, a landmark result known eponymously by LMN \cite{lmn1993ac0}. Subsequent works have led to a tight characterization of this Fourier spectrum \cite{hastad2017average, tal2017fourierac0}. This suggests a natural approach to resolving whether $\PARITY \in \QACZ$ by seeking an analogous characterization for $\QACZ$. The Fourier spectrum of $\QACZ$ circuits, induced by measuring the output qubit, has been conjectured to exhibit low degree concentration  \cite{nadimpalli2024pauli, anshu2025computational}, suggesting a quantum analog of LMN (``QLMN") may hold. Informally, this conjecture states that the $\poly(n)$-level Fourier coefficients have negligible ($n^{-\omega(1)}$) weight. A fully parameterized version of this conjecture, analogous to the true LMN is stated in \Cref{conj:qlmn}.

Clearly, if the QLMN conjecture were true, then $\PARITY \not \in \QACZ$. In this work, we show that the reverse implication also holds:
$$\PARITY \not \in \QACZ \Longleftrightarrow \mathrm{QLMN}.$$
\noindent In other words, proving QLMN is \emph{necessary} for proving $\PARITY \not \in \QACZ$.

Another interpretation of our result is that any single-output $\QACZ$ circuit with non-negligible high-level Fourier mass (\cref{def:booleanckt}) suffices to compute $\PARITY$ in $\QACZ$. This equivalence has no classical analog because adding a gate to $\ACZ$ that only approximates some arbitrary high-degree boolean function, does not necessarily produce $\ACZ[2]$.  

An immediate consequence of our result is that any $\QACZ$ circuit achieving non-negligible correlation with $\PARITY_n$ on a random input provides a $\QACZ$ circuit for computing $\PARITY_n$ exactly. Similarly, it suffices to construct a circuit with $1-1/\poly(n)$ correlation with $\MAJORITY_n$ on average, or a circuit with $1/\poly(n)$ advantage for $\MAJORITY$ in the worst case to obtain exact $\PARITY_n$. These results and others are summarized in \cref{tab:prior} below. Prior work of \cite{grier2024threshold} introduced the worst-case bounded error classes $\bqac^0$, $\bqtc^0$, and \cite{grier2026tc0} subsequently showed $\bqtc^0 = \bqac^0_f$ and $\QTCZ = \QACZ_f$. In regards to these classes, our results imply that if $\MAJORITY \in \bqac^0$, then $\QACZ = \QACZF$.


\paragraph{An average-case separation.} $\MAJORITY$  is often used to illustrate limitations of Fourier-concentration for obtaining average-case $\ACZ$ lower bounds. Despite the Majority function having $1-o(1)$ of its weight on low-degree Fourier coefficients, $\ACZ$ circuits cannot achieve better than $\Theta\lr{\frac{1}{\sqrt{n}}}$-correlation with $\MAJORITY$ \cite{razborov1987lower, smolensky1987algebraic}. We, however, show that a $\QACZ$ circuit can achieve $(1-1/\polylog(n))$ correlation with $\MAJORITY_n$. This establishes the first average-case decision problem separation between $\ACZ$ and $\QACZ$. Furthermore, this sharp contrast with the classical setting suggests that, in the quantum world, Fourier concentration may be largely sufficient to characterize $\QACZ$. Indeed, by combining this result with our aforementioned hardness result, we see that the feasible regimes of $\MAJORITY$ in $\QACZ$ are essentially determined by QLMN.

\paragraph{Consequences for hardness of state synthesis.} It is well-known that preparing the cat state ($\frac{1}{\sqrt{2}}\ket{0^n}+\frac{1}{\sqrt{2}}\ket{1^n}$) or a nekomata state $(\kb{0^n} + \kb{1^n})$ in $\QACZ$ is equivalent to 
computing $\PARITY$ in $\QACZ$. This equivalence also holds for preparing high-fidelity (i.e. $1 - \eps$) approximations to these states \cite{moore1999qac0, rosenthal2021qac0, grier2026tc0}.
However, we show that there are much weaker approximations whose preparation implies $\PARITY \in \QACZ$, but these cannot be distinguished from states known to be in $\QACZ$ through existing measures. To more extensively capture state synthesis tasks that are as hard as \PARITY, we introduce a new notion, \emph{felinity} (\Cref{def:felinity}). The felinity of a $n$-qubit state $\rho$ directly quantifies the ability of a circuit for $\rho$ to correlate with $\PARITY_n$. Combined with our main result, this yields a large class of states whose preparation as black-box immediately unlock the full power of $\QACZF$.


Another family of quantum states of interest is Dicke states, defined as uniform superpositions over $n$-bit strings of Hamming weight $k$ (WLOG $k \leq n/2$). 
Recent work shows that $\FANOUT_n$ enables preparation of arbitrary-weight states in $\QACZF$ \cite{cleandicke2026}. However, it was previously unknown whether this reliance on fanout is necessary. Using felinity, we show that $\FANOUT_k$ is necessary to prepare Dicke states of weight $k$.
Consequently, for any $k=\poly(n)$, such as $k=n^{0.0001}$, the ability to prepare weight-$k$ Dicke states in $\QACZ$ would imply $\QACZ=\QACZF$.

\begin{table}[t!]\label{tab:prior}
\centering
\footnotesize
    \renewcommand{\arraystretch}{1.5}
    \setlength{\tabcolsep}{3pt}
    \begin{tabular}{|c|c|c|c|c|c|}
\hline
\cellcolor{gray!10}Problem &   \cellcolor{gray!10} Regime &  \cellcolor{gray!10}  $\QACZ$  & \cellcolor{gray!10}  Implies $\qaczm = \qaczm_f$ &  \cellcolor{gray!10}  $\ACZ$ & \cellcolor{gray!10} $\QACZF$\\
\hline 
\multirow{3}{*}{\shortstack{\PARITY$_n$ \\ (Avg correlation)}} 
    &  Exact &  ? & Yes \cite{moore1999qac0} & \multirow{3}{*}{{No \cite{hastad1986switch}}} &  \multirow{3}{*}{{Yes \cite{moore1999qac0} }}   \\
\cline{2-4}
&  $\geq 1/\polylog(n)$ & ? & Yes \cite{joshi2026improvedlowerboundsqac0} &  &  \\
\cline{2-4}
&  {$\geq 1/\poly(n)$ } & ? &  {\bf Yes (Thm \ref{thm:fanoutcomplete_intro}/\ref{thm:nonnegpar_to_exact}})\cellcolor{green!15} & &  \\
\hline 
\hline 
\multirow{3}{*}{\shortstack{\MAJORITY$_n$ \\ (Worst Case \emph{error})  }} 
&  $1/2 - \Theta(1/n)$ & Yes \cite{grier2026tc0} & -  & Yes*&  \multirow{3}{*}{Yes \cite{hoyer2005fanout} } \\
\cline{2-5}
&  0  & ? &  Yes \cite{grier2024threshold}  & \multirow{2}{*}{No \cite{hastad1986switch}} &   \\
\cline{2-4}
&  $1/2-1/\polylog(n)$ & ? & \cellcolor{green!15} {\bf Yes (Thm \ref{thm:fanoutcomplete_intro}/\ref{thm:hardmajority} ) } &   &   \\
\hline
\multirow{2}{*}{\shortstack{\MAJORITY$_n$ \\ (Avg correlation)  }} 
&  $1-1/\polylog(n)$& \cellcolor{green!15} {\bf Yes (Thm \ref{thm:ezmajority} ) } &  -  &\multirow{2}{*}{\shortstack{No \cite{razborov1987lower}, \\\cite{smolensky1987algebraic}}}  & \multirow{2}{*}{Yes \cite{hoyer2005fanout}}   \\
\cline{2-4}
&  $1-1/\poly(n)$& ? & \cellcolor{green!15} {\bf Yes (Thm \ref{thm:fanoutcomplete_intro}/\ref{thm:hardmajority} ) } &   &   \\
\hline
$\MAJORITY_n \land \PERIODIC$ 
& \multirow{2}{*}{$\negl(n)$} &  \multirow{2}{*}{Yes \cite{grier2024threshold}} & \multirow{2}{*}\shortstack{-}
& \multirow{2}{*}{No \cite{hastad1986switch}} & \multirow{2}{*}{Yes \cite{hoyer2005fanout}}   \\
 (Worst-case \emph{error}) & & & & & \\
\hline
\hline
\multirow{3}{*}{\shortstack{Nekomata \\ $\braket{0^n | \rho | 0^n}, \braket{1^n | \rho | 1^n} \geq \eps$}} 
    &  $\eps = 1/2$ &  ? & Yes \cite{rosenthal2021qac0} & \multirow{3}{*}{{-}} &  \multirow{3}{*}{{Yes \cite{moore1999qac0} }}   \\
\cline{2-4}
&   $\eps \approx 1/\sqrt{2}$   & ?  &  Yes \cite{rosenthal2021qac0,grier2026tc0} & &  \\
\cline{2-4}
& {$\eps \geq 1/\poly(n)$} & ?  & \cellcolor{green!15} {\bf Yes (Thm \ref{thm:fanoutcomplete_intro}/\ref{cor:felnhard}} & &  \\
\hline
\hline
\multirow{2}{*}{\shortstack{Dicke states: $\ket{D^n_k}$ } } 
    &  $k = n^{\Theta(1)}$ &  ? &\cellcolor{green!15} {\bf Yes (Thm \ref{thm:fanoutcomplete_intro}/\ref{thm:dicketocat})} & \multirow{2}{*}{-} &  \multirow{2}{*}{Yes \cite{cleandicke2026}}   \\
\cline{2-4}
&  $k = O(1)$    & Yes \cite{cleandicke2026} & - & &  \\
\cline{2-4}
\hline
\end{tabular}
\caption{Prior and new results regarding the hardness of problems related to $\FANOUT$. Prior work is elaborated upon in \Cref{sec:prior_work}. New results from our work are highlighted in green. Here worst case error indicates that the error guarantee holds for all inputs. The worst case error is of form $1/2-\eps$ since error $1/2$ is trivially achieved by a random bit. Average-case correlation is taken in expectation over a uniformly random input $\x \in \bin^n$. \\ *This is possible with access to $\log(n)$ random bits.}  
\end{table}

\subsection{Our results}

We show that the Fourier concentration of $\QAC$ functions is intimately related to the hardness of $\PARITY$ in $\QAC$. More formally, define $\QAC(d, s)$ to be $\QAC$ circuits with depth $\leq d$ and size $\leq s$\footnote{There is some subtlety based on whether $\QACZ$ is defined to have polynomial ancillae or polynomial gates. In this paper this concern is largely not important, as our reductions achieve proportional blowup in both measures.} and define the $\text{QLMN}(\gamma(d,k))$ to be the following proposition, parameterized based on a decay function $\gamma : \mathbb{N} \times \mathbb{N} \to \mathbb{N}$, similar to the way the classical \cite{lmn1993ac0} is parameterized.

\begin{restatable}[$\text{QLMN}(\gamma)$ (folklore, \cite{nadimpalli2024pauli, anshu2025computational})]{conjecture}{qlmn}\label{conj:qlmn}
For any $d,s, n$ and all $n$-input $C \in \QAC(d,s)$ with a designated output register $t$ and let $f_C(x) : \{0,1\}^n \to [-1,1]$ denote the expected outcome of measuring $Z_t$ on $t$. Then, for all $k \in [n]$. 
    \[ \wgk\left[f_\mathcal{C}\right] < s \cdot \gamma(d,k) ^ {-1}.\]
\end{restatable}
Note that setting $\gamma(d,k) := 2^{\Omega\left(k^{1/d}\right)}$ would conjecture the analog of classical LMN. We show that proving strong lower bounds on the size/depth complexity of $\PARITY$ in $\QAC$ necessitates proving QLMN under corresponding parameters.

\begin{restatable}[Equivalence of QLMN and Hardness of Parity]{theorem}{qlmnequivparity}\label{thm:qlmnequivparity}

For all $\gamma : \mathbb{N} \times \mathbb{N} \to \mathbb{N}$, \[\mathrm{QLMN}(\gamma) \iff \PARITY_n \not\in \QAC(\Theta(d), \wt{\Theta}(\gamma(d,n))).\]

Equivalently, letting $S(d,n)$ denote the maximum size such that no depth $d$ circuit of this size can compute $\PARITY_n$, then $\text{QLMN}(\tilde{\Theta}\left(S(\Theta(d), k))\right)$ holds.
\end{restatable}
Here $\Theta$ hides absolute constants, and $\tilde{\Theta}$ hides absolute constants and log factors. Note that the dependence on circuit size on the LHS is baked into the $\text{QLMN}$ definition. Restricting to the case of $\QACZ$, we have the following corollary.

\begin{corollary}
$\text{QLMN}(\poly_d(k^{\omega(1)}))$ holds for some $\poly_d(k^{\omega(1)})$ if and only if $\PARITY \notin \QACZ$.
\end{corollary}

Next, we provide that a variety of approximation problems and state construction tasks that are $\QACZF$-complete, or equivalently $\QTCZ$-complete, under $\QACZ$ reductions. Augmenting $\QACZ$ gates with a ``noisy" gate (and its inverse) for any one of the following tasks, provides the ability to implement any problem in $\QACZF$ \emph{without error}.

\begin{theorem}[$\QACZF$-complete problems]\label{thm:fanoutcomplete_intro}
The following tasks are equivalent to exact \FANOUT$_n$  under $\QACZ$ reductions.
\begin{enumerate}
    \item \textbf{Any non-negligible \PARITY$_n$ correlation} (\cref{thm:nonnegpar_to_exact}). Achieving non-negligible (i.e. $\geq 1/\poly(n)$) correlation with $\PARITY_n$ on a random input.
    \item \textbf{Weak nekomata-like states} (\cref{cor:felnhard}). For any non-negl $\eps \geq 0$, synthesizing a state $\rho$ satisfying $\braket{0^n | \rho |0^n} \geq \eps, \braket{1^n |\rho | 1^n} \geq \eps$.
    \item \textbf{Poly-weight Dicke states} (\cref{thm:dicketocat}). For any $\delta \in (0,1)$ and $k = \Omega(n^{\delta}) \leq n/2$, synthesizing a state with at least $(1-\eps)$ fidelity with $\dkn$. 
\item \textbf{High average case \MAJORITY$_n$ correlation} (\cref{thm:hardmajority}). For any arbitrarily small constant $\delta > 0$, achieving  at least $(1-1/n^{\delta})$  correlation with $\MAJORITY_n$ on a random input. 
    \item \textbf{Bounded worst case \MAJORITY$_n$ error} (\cref{cor:maj_worstcase}). For any $\eps= 1/\polylog(n)$, 
    achieving $\eps$-correlation with $\MAJORITY_n(\x)$ (correct wp $\geq 1/2 + \eps/2$) on every input $\x \in \bin^n$. 
\end{enumerate}
In other words, a $\QACZ$ circuit $C$, whose $m = \poly(n)$ ancillae start in $\ket{0^m}$, for any one of the above tasks implies $\qaczm = \qaczm_f$. 
\end{theorem}
Notice that ``weak-nekomata-like states'' can have much lower fidelity and trace-distance to any nekomata compared to states trivially in $\QACZ$, such as $\ket{0^n}$.  Moreover for small enough $\eps_0, \eps_1$, they have much lower \emph{mutual information} between qubits than the $\ket{W_n}$ state (weight $1$ Dicke state), which can be synthesized exactly in $\QACZ$ \cite{grier2026tc0, cleandicke2026}. 
To reliably discern these hard states, we introduce a new measure, \emph{felinity} (\cref{sec:felinity}), and show that preparing any state with non-negligible felinity implies $\QACZ = \QACZF$ (\cref{cor:felnhard}).

Finally, we crucially distinguish between average and worst case regimes for $\MAJORITY$.
This is in contrast to $\PARITY_n$, where average case and worst case error regimes for $\PARITY$ are equivalent for $\QACZ$ as well as $\ACZ$ (with randomness). 
Both $\QACZ$ and $\ACZ$ can achieve $\Theta(1/\sqrt{n})$-correlation with $\MAJORITY$, i.e., $1/2 + \Theta(1/\sqrt{n})$ success probability, on a random input by outputting the first input bit. It is also known that $\ACZ$ cannot achieve better correlation than this for $\MAJORITY$. However, our next result shows that this is not true for $\QACZ$, thus establishing an average-case separation between the two classes.

\begin{theorem}[\cref{thm:ezmajority} informal]\label{thm:ezmaj_intro}
There exists a $\QACZ$ circuit that achieves any $(1-1/\polylog(n))$ correlation with $\MAJORITY$.
\end{theorem}

This brings us back to Moore's original question of whether \PARITY$_n \in$ $\QACZ$. This question is indeed the key to determining the computational power of $\QACZ$, regardless of which notion of computation, exact/clean/approximate/average case/worst case, is used. Our results clear the path towards understanding $\QACZ$ and motivate further study into the Fourier spectrum of $\qaczm$. From a practical standpoint, it is sufficient to produce a gate for one of the much weaker approximation tasks from \cref{thm:fanoutcomplete_intro} to unlock the full power of $\QACZF$, which includes modular arithmetic, QFT, sorting, and arbitrary Dicke state preparation \cite{moore1999qac0, green2002acc, hoyer2005fanout, takahashi2012collapse, cleandicke2026}. 


\subsection{Prior work}\label{sec:prior_work}

A key distinction between $\QACZ$ and $\ACZ$ circuits, arising from reversibility, is the role of \FANOUT. In $\ACZ$, copying the outputs of gates is free. Quantumly, due to the no-cloning theorem, copying the output of gates is impossible. Instead, $\FANOUT$ refers to copying the classical information of a qubit. The \emph{reversible} $\FANOUT_n$ gate copies a qubit onto $n$ targets via XOR, and under Hadamard conjugation corresponds to a clean $\PARITY_n$ gate. 
Foundational work by \cite{moore1999qac0} established equivalences between $\FANOUT_n$, $\PARITY_n$, and cat/GHZ states $(\ket{\Cat_n} = \frac{1}{\sqrt{2}} \ket{0^n} + \frac{1}{\sqrt{2}}\ket{1^n})$. These can be thought of as complete problems under $\QACZ$ reductions for the class obtained by including fan-out, $\QACZF$. Early works also showed that $\QACZF$ can implement arbitrary threshold functions, modular arithmetic, and more \cite{moore1999qac0, green2002acc, hoyer2005fanout}. Subsequently, \cite{rosenthal2021qac0} discovered that preparing a ``nekomata", a weaker form of $\ket{\Cat_n}$ consisting of a \emph{mixture} rather than a superposition,  is also $\QACZF$-complete. They leveraged this new reduction to provide the first constant depth approximation of $\PARITY_n$ using exponential size. This allowed for approximate-$\FANOUT_{\polylog(n)}$ in $\QACZ$, which is limited to $\poly(n)$ size. Thereafter, 
\cite{grier2024threshold} showed that threshold gates also provide the ability to create approximate nekomata. Recently both constructions were made exact by \cite{grier2026tc0}, implying that $\QTCZ = \QACZF$ and $\FANOUT_{\polylog(n)}$ can be exactly implemented in $\QACZ$. Existing prior work on upper bounds and reductions is summarized in \cref{tab:prior}.

Prior work such as \cite{grier2024threshold, rosenthal2021qac0} has focused mainly on worst-case error regimes, where performance is measured by the \emph{error} guarantee for all inputs. \cite{rosenthal2021qac0} showed reductions between the worst-case approximations of $\QACZF$-complete problems. For high accuracy, such as $(1-o(1))$ fidelity with a nekomata or $o(1)$ worst-case error for parity, \cite{grier2026tc0}'s result implies $\QACZF$-completeness by reducing from exact nekomata. 

We also consider average-case regimes in \cref{tab:prior}, which are more natural from a Fourier-analytic standpoint. Average-case performance is measured by \emph{correlation} with the target function on a uniformly random input. For $\PARITY_n$, average-case error can be converted to worst-case error by mapping the input to a random string of the same parity, a transformation implementable in both $\QACZ$ \cite{watts2019separation} and $\ACZ$\footnote{provided access to random bits.}.  
Analogous to the amplification possible in $\ACZ$ via $\MAJORITY$ of parallel repetition, \cite{joshi2026improvedlowerboundsqac0} showed that $1/\polylog(n)$-correlation with $\PARITY_n$ can be amplified to exact $\PARITY_n$.

Regarding upper bounds, known capabilities of $\QACZ$ have been limited. Existing constructions include regimes of parity, thresholds and symmetric functions matching $\ACZ$ \cite{rosenthal2021qac0, grier2024threshold}. In addition, \cite{cleandicke2026} show that $\QACZ$ can prepare constant-weight $n$-qubit Dicke states. It is unknown whether $\QACZ$ contains all decision problems in $\ACZ$. The counterpart of this question, whether $\QACZ$ contains problems outside $\ACZ$ has been addressed to some degree.
Earlier works, \cite{Bravyi_2018, watts2019separation} established relational separations between the two classes. Recently, \cite{grier2026tc0} was the first to establish a worst-case decision problem separating $\QACZ$ and $\ACZ$.  
They showed that $\qaczm \circ \NC^0$ contains $\TCZ$, a strictly more powerful class than $\ACZ$, demonstrating that even limited classical fanout (via multiple copies) significantly increases the power of $\QACZ$. The one-sided error in \cite{grier2026tc0}'s construction is not achievable by $\ACZ$ and forms the basis for their separation. 

\section{Technical Overview}
We provide an overview of all our proof techniques organized as follows. In \cref{sec:ovpar}, we describe the equivalence to QLMN and the corollary leading to point (1) of \cref{thm:fanoutcomplete_intro}. Next, in \cref{sec:ovstate} we discuss the hardness of state preparations tasks which cover points (2,3) of \cref{thm:fanoutcomplete_intro}. Finally, in \cref{sec:ovmaj} we focus on $\MAJORITY$ and discuss points (4,5) as well as the upper bound in \cref{thm:ezmaj_intro}. The full proofs are split into two sections, reductions (\cref{sec:qlmn}) and upper bounds (\cref{sec:upper}).

\subsection{Characterizing \safeqaczf~Completeness}\label{sec:ovpar}
It is possible in $\QACZ$ to obtain a negligible but non-zero correlation with $\PARITY_n$  by hard-coding the answer to a small subset of the possible inputs. It is also possible to construct a state that has non-zero but negligible amplitude on the lower-weighted nekomata branch \cite{rosenthal2021qac0}. We show that any meaningful improvement over these approximations is as hard as exact parity.

\paragraph{Challenges and comparison to prior work.}
Existing approaches for reducing from exact $\PARITY$ to approximate parity/nekomata regimes require sufficiently strong approximations. The ``classical" approach involves first mapping the input to a uniformly random string of the same parity and amplifying by taking the $\MAJORITY$ of parallel runs. Since $\MAJORITY_{\polylog(n)}$ can be computed in $\poly(n)$ size, this applies when the correlation with $\PARITY_n$ is at least $1/\polylog(n)$ \cite{joshi2026improvedlowerboundsqac0}. In contrast, the techniques we use to prove points (1,2,3) of \cref{thm:fanoutcomplete_intro} do not have classical analogues. This is because, arbitrary $1/\poly(n)$ correlations with $\PARITY_n$ cannot be black-box amplified in $\ACZ$ since they require $\MAJORITY_n$. 

Prior inherently quantum approaches work by amplifying the weight on nekomata branches, but still require sufficient weight on both branches, $\braket{0^n | \rho | 0^n}$ and $\braket{1^n | \rho | 1^n}$. Rosenthal’s nekomata amplification \cite{rosenthal2021qac0} applies to states with most of their weight on one nekomata branch and small but non-negligible weight on the other. The amplitude amplification techniques of \cite{grier2026tc0}, require a total constant weight on the two branches to amplify to the ``good subspace" (supported only on $\ket{0^n}, \ket{1^n}$) in constant depth.
Consequently, these methods do not apply to the regimes we consider, where the only guarantee is non-negligible weight on two branches. This forces us to dig deeper to pinpoint the additional structure in these approximation tasks.

\subsubsection{Fourier concentration to Parity}
Any $n$-input quantum circuit $C$ with a designated output register has an associated real-valued Boolean function, $f_C : \bin^n \to [-1,1]$. This function is obtained by standard basis measurement distribution of the output qubit, or equivalently by the expected outcome on a $Z$ observable. 
For such functions $f$ and $0 \leq k \leq n$, the weight above level $k$, $\wgk[f]$ is defined as \cite{o2014analysis}, 
\begin{align}
    \wgk[f] := \sum_{\substack{S \subseteq [n] \\ |S| \geq k}} \wh{f}(S)^2.
\end{align}

\noindent This defines the Fourier mass of a $n$-input $\QAC$ circuit $C$ naturally through $\wh{f_C}$.  

Recall the definition of the QLMN proposition.

\qlmn*

\noindent Note that $\PARITY_n$ does not have this type of concentration for any reasonable $\gamma$, since all the mass is on the degree $n$ coefficient. Thus, the non-trivial direction is to show that a $\QACZ$ circuit $C$ for which $f_C$ is non-concentrated produces a $\QACZ$ circuit for exact $\PARITY_n$. Formally, we show the following. 

\qlmnequivparity*

This immediately provides the following consequences for approximate parity, thus proving point (1) of \cref{thm:fanoutcomplete_intro}.

\begin{restatable*}[Non-negligible correlation with parity implies fanout]{corollary}{approxtoexactpar}\label{thm:nonnegpar_to_exact}
Let $C$ be a depth $d$ $\QACZ$ circuit $C$ using $m$ ancillas whose designated output register achieves at least $\eps$ correlation with $\PARITY_n$. Then, there exists a $\QACZ$ circuit $C'$ of depth $d' = O(d)$ using $m' = \tilde{O}((m+n)/\eps)$ qubits that implements $\FANOUT_n$. 
\end{restatable*}


The key step for this proof is a gadget which leverages $C$ to construct a ``nekomata-like'' state that has constant amplitude on the $\ket{0^n}$ branch and non-negligible amplitude on the subspace of strings with hamming weight at least $k$. From there, the following lemma suffices:


\begin{restatable*}[Threshold-block nekomata is $\QACZF$-complete]{lemma}{twobranchdicke}\label{lem:two_branch_dicke}
Let $C$ be a depth $d$, $m$ ancilla circuit that constructs  $\ket{\psi}_{Q,A}$ with $|Q| = n$, satisfying,
$$\Vlr{\bra{0^n}_Q \cdot \ket{\psi}}^2 \geq 1/2 \quad \text{and} \quad \Vlr{{\pnk}_Q \cdot \ket{\psi}}^2 \geq \eps$$
where $\pnk$ is the projector onto $n$-bit strings of Hamming weight at least $k$,  for $0 < k \leq n$.  
Then, there exists a $\QACZ$ circuit of depth $d' = O(d)$ and $m' = O((m+n)/\eps)$ that exactly implements $\FANOUT_{\ell}$ for $\ell = k/\log^2 k$. 
\end{restatable*}
To prove this lemma, the idea is to partition the qubits into $\ell$ blocks, and then apply an $\OR$ gate onto a fresh ancilla. This partitioning is chosen at random. Next, we use amplitude amplification to remove the ``bad" part which consists of the subspace orthogonal to $\kb{0^n} + \kb{1^n}$. This produces a state that has most of its weight on one nekomata branch and a small but non-negligible weight on the other and can be amplified to an exact nekomata through a combination of parallel and amplitude amplification. Through \cite{rosenthal2021qac0} this provides a circuit for fanout/parity. Furthermore, when $k = n^{\delta}$, this also provides a circuit for $\FANOUT_n$ with $1/\delta$ factor increase in depth.  

We now state our key step. 
\begin{restatable*}[Fourier mass to Hamming slice]{lemma}{threshbranchwt}\label{lem:qlmntothres}
For any depth $d$, $m$-ancilla $\QACZ$ circuit $C$ acting on $n$ inputs, there exists a depth $O(d)$, $m' = m+2n+1$-ancilla $\QACZ$ circuit $C^*$ that constructs a state $\ket{\psi^*}$ satisfying 
$$\Vlr{ \bra{0^n} \cdot \ket{\psi^*}}^2 \geq 1/2 \quad \text{and} \quad \Vlr{\pnk \cdot \ket{\psi^*}}^2 \geq \wgk[f_C]/2.$$
for any $0 < k \leq n$. Here $\pnk$ is the projector onto $n$-bit strings of Hamming weight at least $k$.  \end{restatable*}
First, given $C$, we obtain an ``approximate clean computation" circuit $C'$ by using the standard copy-uncompute technique. This is done by first making a classical copy to preserve the input registers, then running $C$ and applying a $\cnot$ gate on the output register and a fresh ancilla, and finally, ``uncomputing" by applying all the previous gates in reverse ($C^\dag$ and the initial copy layer). When the output of $C$ is boolean valued, $f_C(\x) \in \bpm$, this transformation produces a clean computation of $f_C$, i.e,
\begin{align}\label{eq:here1}
    \ket{\x}_T \ket{0^{m'}} \mapsto_{C'} \ket{\x}_T \ket{f_C(\x)} \ket{0^{m'-1}}
\end{align}
(where we interpret $\ket{f_C(\x)}$ as $\in \clr{\ket{0}, \ket{1}}$ based on its $Z$ eigenvalue).
Although \cref{eq:here1} does not hold when the output qubit is non-deterministic, with $f_C(\x) \in [-1,1]$, we can still bound the probability of the ancillae being returned to $\ket{0^{m-1}}$ on a given input. We use this observation to prove the lemma for the state $\ket{\psi^*} := H^{\tens n}_T C' \ket{+^{n}}_T \ket{0^{m'}}$ by post-selecting on the subspace containing $\ket{0^{m-1}}$ on the non-output ancillae. 

We show that the claimed overlap with $\ket{0^n}_T$ follows simply by virtue of $C$ being a single-output circuit, where, the variance of the measured output is bounded by $1/2$. To argue about the amplitude in the subspace corresponding to $\pnk$, we construct a special $+1$ eigenvector of $\pnk$, $\ket{\T_k}$, based on the Fourier coefficients of $f_C$. This $\ket{\T_k(C)}$ is essentially a gadget we design to extract the Fourier mass of $f_C$ through the bilinear form, 
\begin{align} 
\lr{\bra{\T_k} \tens \bra{-} \tens \bra{0^{m'-1}}} \ket{\psi^*_C} = \wgk[f_C]/2.
\end{align}
Since the overlap of $\ket{\psi^*}$ with any $+1$ eigenvector of $\lr{\pnk \tens I}$ lower bounds the quantity $\Vlr{\pnk \cdot \ket{\psi^*}}$, this provides us the claimed bound.

\subsection{Felinity as a measure of state preparation complexity}\label{sec:ovstate}
Existing notions of state approximation are insufficient to characterize states that are as hard as cat states in $\QACZ$. Taking fidelity for example, $\braket{\Cat_n | 0^n}^2 = 1/2$, so it is not the case that all states of constant fidelity with $\ket{\Cat_n}$ are hard. Fidelity with cat states/nekomata is only meaningful in this context when it is bounded above $1/2$. 

Another natural notion is based on mutual information between pairs of qubits. 
A nekomata can be thought of as a mixture over $n$ perfectly correlated classical bits. Thus, both nekomatas and Cat states have mutual information $1$ between every pair of the $n$ targets. On the other hand, there are states easily preparable in $\QACZ$, such as $\ket{W_n}$, with pairwise mutual information $1/n$. For a mutual-information based state hardness metric, this suggests a high cutoff, eg $\Omega(1/\polylog(n))$. We show that states with much smaller pairwise mutual information, such as $\eps \ket{0^n} + \eps \ket{1^n} + \ket{\perp}$, and much lower fidelity with \emph{any} nekomata are just as hard to construct as a nekomata in $\QACZ$. 

The same issues arise when using other related metrics such as trace distance, \emph{phase-dependent-fidelity} \cite{rosenthal2021qac0}, or \emph{interaction information}. We propose a new metric to quantify the complexity of state preparation tasks, \emph{felinity}, parameterized by the number of targets $n$. 

\begin{restatable}[Felinity of a quantum state]{definition}{felinitydef}\label{def:felinity}
 For a $n$-qubit state $\rho$, the \emph{felinity} of $\rho$, $\feln(\rho) \in [0,1]$ is defined as
 $$\feln(\rho) := 2 \sum_{\y \in \bin^n} \braket{\y|\rho|\y} \cdot \braket{\y | X^{\tens n} \rho X^{\tens n} | \y}$$
\end{restatable}

\noindent  Felinity has several natural properties, such as \emph{monotonicity}, for any qubit $q$ of $\rho$, 
\begin{align}
\feln(\rho) \leq \fel{n-1}(\tr_q \rho), 
\end{align}
and \emph{Lipschitz continuity} (\cref{lem:feltd}),
\begin{align}
|\feln(\rho)-\feln(\sigma)| \le 8\,\TD(\rho,\sigma). 
\end{align}
Felinity behaves as expected on the states discussed so far, 
\begin{itemize}
\item Any perfect $n$-nekomata $\rho$ has $\feln(\rho) = 1$, the highest possible, and any $\rho'$ that is $\eps$-close in trace distance to some nekomata has $\feln(\rho') \geq 1-\Theta(\eps)$.  
\item For any $n$-qubit product state, $\rho = \rho_1 \tens \rho_2 \dots \rho_n$, $\feln(\rho) \leq 2^{-n+1} = \negl(n)$.
\item Any $n$-qubit state $\rho$ constructed by a $\QNC$ circuit has $\feln(\rho) = \negl(n)$ (light-cone argument). 
\item  $\feln(W_n) = 0$, and in a different basis, any $\ket{\psi} = U^{\tens n} \ket{W_n}$ has $\feln(\psi) \leq \negl(n)$ (See \cref{lem:felwstate}) 
\item Any state with non-negligible mass on a string $\ket{\y}$ and the complement string $\ket{\ov{\y}}$, such as point (2) of \cref{thm:fanoutcomplete_intro}, has non-negligible felinity. 
\end{itemize}
We demonstrate that \emph{felinity} captures $\QACZF$-completeness of state preparation tasks by showing that a $\QACZ$ circuit for constructing a $n$-qubit state with non-negligible felinity implies $\PARITY_n \in \QACZ$. The key step is the following lemma, 
\begin{restatable*}[Felinity gives advantage for parity]{lemma}{parfromfeln}\label{lem:felinity_to_par}
 Let $\ket{\psi} = C \ket{0^m}$ be a state on $m = \poly(n)$ qubits constructed by a depth $d$ $\QACZ$ circuit $C$. Suppose that there exists a subset $T$ of $|T| = n$ qubits for which the state $\rho_T = \tr_{\ov{T}} \kb{\psi}$ has $\feln(\rho_T) \geq \eps$. Then, there is a depth $d' = 2d+4$ circuit $C'$ on $m' = m+2$ ancillas whose output achieves $\geq |\eps|$ correlation with $\PARITY_n$.
\end{restatable*}
\noindent We first recall the idea behind the Moore's original reduction to Cat state \cite{moore1999qac0, rosenthal2021qac0}. For any nekomata $\ket{\nu_n}_{T,A}$ with $n$ targets $T$ and ancillae $A$, we have the following property. Let $L$ be a layer of $n$ $\cz$ gates each acting on an input qubit and a corresponding target from $T$. On any input $\x$, the state $L \ket{\nu_n}_{T,A} \ket{\x}$ is orthogonal to $\ket{\nu_n}_{T,A}$ iff $\x$ has odd parity. To extract this information, we apply the reflection $R = (I - 2\kb{\nu_n}_{T,A} \tens \kb{+}_q)$ on a fresh output register $q$, which acts as a controlled-$X_q$ operator, controlled on $T,A$ being $\ket{\nu_n}$. $R$ can be implemented using the circuit $C$ for $\ket{\nu_n}$, though, 
\begin{align} 
R = C^\dag \cdot (I - 2\kb{0^m} \tens \kb{+}_q) \cdot C
\end{align}
This provides a single-output circuit $C' = R \cdot L$ with an output register $q$ for exactly computing $\PARITY_n$. Back to the general case of \cref{lem:felinity_to_par}, where $C$ constructs an arbitrary state $\ket{\psi}$, we start by defining $C'$ the same way as above. Then, the focus of our proof is to relate the Fourier coefficients of $f_{C'}$ to $\ket{\psi}$. 

Previous works generally treat the output of the circuit as a qubit or through the probability of measuring $\ket{1}$ on the output register. We adopt a different but equivalent \emph{observables} based viewpoint.
Specifically, $f_{C'}(\x)$ as we defined in \cref{def:booleanckt} captures the expected $\pm 1$ outcome of measuring the $Z$ observable on the output register $q$. Additionally, any Hermitian matrix $\O$ is a valid observable corresponding to a measurement with possible outcomes among the eigenvalues of $\O$. Thus, we can associate the following $2^n \times 2^n$ observable with $C'$ by contracting the ancillae, 
 \begin{align}
 \O_{C'} := \bra{0^m}_A \tens I \cdot \lr{C^\dag \cdot Z_t \tens I \cdot  C} \cdot \ket{0^m}_A \tens I
 \end{align}
This is essentially the Heisenberg evolution of the output measurement with eigenvalues in $[-1,1]$. Therefore, $\braket{\x | \O_{C'} | \x} = f_{C'}(\x)$. This re-interpretation of $f_{C'}$ allows us to cleanly expand its Fourier coefficients in terms of the quantities involved in \cref{def:felinity} for $\rho_T$, revealing that, 
\begin{align}
\wh{f_{C'}}([n]) = \feln(\rho_T).
\end{align}
\noindent and hence $\corr(f_{C'}, \PARITY_n) \geq \eps$. Then, \cref{thm:nonnegpar_to_exact} completes the reduction from exact $\PARITY_n$, providing the following corollary.  
\begin{restatable*}[Non-negl felinity states are $\QACZF$-complete]{corollary}{felntofanout}\label{cor:felnhard}
Let $\rho$ be any $n$-qubit state such that $\feln(\rho) = \eps$. If $\rho$ can be constructed by a depth $d$, $m = \poly(n)$ ancilla circuit $C$, then there is a depth $d' = O(d)$, $m' = \tilde{O}((m+n)/\eps)$ ancilla circuit to implement $\FANOUT_n$.
\end{restatable*}
\noindent As a consequence, preparing any non-negligible felinity state on $n$ qubits in $\QACZ$, such as states in point (2) of \cref{thm:fanoutcomplete_intro}, implies $\qaczm = \QACZF$. 

\subsubsection{Dicke states imply parity}
To demonstrate an application of \emph{felinity}, we show that high-weight Dicke states can be turned into high-felinity states through a single layer of $\QACZ$ gates. When discussing Dicke states, $\ket{D^n_k}$, we without loss of generality consider the case when the weight $k \leq n/2$. This is because any $\ket{D^{n}_{n-k}}$ can be obtained by negating the bits of $\ket{D^n_k}$.
This allows us to quantify the hardness of Dicke states as a function of $k$, and the $\QACZ$-complete states as those with $k = n^{\delta}$ for any constant $\delta \in (0,1)$.  Then, we state the full version of point (3) of \cref{thm:fanoutcomplete_intro} below. 

\begin{restatable*}[Dicke states imply Fanout]{theorem}{dicketocat}\label{thm:dicketocat}
Suppose there exists a depth $d$ $\QAC$ circuit $C$ using $m$ ancillas that, for some $k \leq n/2$, constructs a $n$-qubit state $\rho$ satisfying, 
$$\TD\left(\rho,\dkn\right) \leq \frac{1}{80k},$$
 Then, there is a depth $d' = O(d)$, $m' = m \cdot \poly(n)$-ancilla $\QAC$ circuit $C'$ that implements $\FANOUT_k$. 
 Additionally, for $\delta = \log k/\log n$, there is a depth-$O(d'/\delta)$, $O(m' n^{1-\delta})$ ancillae $\QAC$ circuit for $\FANOUT_n$. 
\end{restatable*}

To prove \cref{thm:dicketocat}, the reduction of \cref{lem:two_branch_dicke}, which only uses classical gates does not suffice. Our key observation here is that $\dkn$ has non-negligible overlap with the states $\ket{p}^{\tens n}$ and $(Z\ket{p})^{\tens n}$, for $p = k/n$. Although $\ket{p}$ and $Z \ket{p}$ are far from orthogonal, for $t \gg n/k$, $\ket{p}^{\tens t}$ and $(Z\ket{p})^{\tens t}$, have negligible overlap. Then, we show that $\ket{D^n_k}$ can be transformed into a state $\rho$ with non-negligible $\fel{\ell}(\rho)$, by applying a layer of $\ell \approx k/\log k$ reflection gates, each acting as a controlled-not onto a fresh ancilla, with controls on $\ket{\rho}^{\tens t}$. This reduces us to \cref{cor:felnhard}. Then, when $k = \poly(n)$, $\delta$ is a constant, implying the following. 

\begin{restatable*}[Preparing poly-weight Dicke state is $\QACZF$-complete]{corollary}{polydickeisfanout}
For any $k = n^{\Theta(1)}$, $k \leq n/2$, if there is a $\QACZ$ circuit that prepares a $n$-qubit state $\rho$ within $1/k$ trace distance of either $\ket{D^n_k}$ or $\ket{D^n_{n-k}}$, then $\QACZ = \QACZF$.
\end{restatable*}


Contrary to the equivalence between cat states and nekomatas, this reduction does not apply to an entangled version of a Dicke state, i.e a uniform \emph{mixture} over strings of hamming weight $k$, rather than a superposition. Perhaps a more appropriate analogue of such a state is the ``Poor-man's cat state" which is a \emph{mixture} over strings of odd parity and can be prepared in $\QACZ$, and in fact in $\QNC^0$ \cite{watts2019separation}. As expected, the poor man's cat state has $\negl(n)$ felinity. The state $\ket{\Cat_n}$ on the other hand, is a \emph{superposition} over even parity strings (in the Hadamard basis), which is crucial to making its felinity $1$.

In conclusion, \emph{felinity} provides a concrete characterization of states that are hard to construct in $\QACZ$ and may be of interest to subsequent work on understanding the power of $\QACZ$ through the ancilla-only state synthesis setting.

\subsection{Hardness classification of Majority regimes}\label{sec:ovmaj}
$\MAJORITY$ is the canonical example used to illustrate the insufficiency of Fourier concentration lower bounds for $\ACZ$. In contrast, we show the feasible regimes of approximate  $\MAJORITY$ in $\QACZ$ are fully determined by their Fourier concentration. First we obtain the following two corollaries as a consequence of \cref{thm:qlmnequivparity}, that show the remaining two points (4,5) of \cref{thm:fanoutcomplete_intro}. 
\begin{restatable*}[High correlation with $\MAJORITY$ implies $\FANOUT$]{corollary}{hardmajority}\label{thm:hardmajority}
  If there is a depth $d$, $m$-ancilla $\QACZ$ whose output achieves at least $(1-\frac{1}{n^{\delta}})$-correlation with $\MAJORITY_n$ for any constant $\delta > 0$, then, there exists a depth $d' = O(d)$, $m' = \poly(m+n)$-ancilla circuit that exactly implements $\FANOUT_n$.
\end{restatable*}

Then, the result follows directly from \Cref{cor:qlmntoparity}.
Next, since a worst case error of up to $1/\poly(n)$ for any function can be boosted to arbitrary $(1-1/\poly(n))$-correlation through parallel repetitions, we also have the following.

\begin{restatable*}[Worst-case approx $\MAJORITY$ implies $\FANOUT$]{corollary}{hardworstmajority}\label{cor:maj_worstcase}
Suppose there is a $n$-input $\qaczm$ circuit $C$ such that for all inputs $\x \in \bin^n$, measuring the output register of $C$ results in $\MAJ_n(\x)$ with probability at least $1/2 + \eps$, for some $\eps \geq 1/\polylog(n)$. 
Then, $\PARITY_n \in \qaczm$. 
\end{restatable*}

\noindent As a consequence, if $\MAJ_n \in \bqac^0$, i.e, can be computed with worst case error $1/3$,  then, $\qaczm = \qaczm_f$.

Analogous $\ACZ$ lower bounds to both these corollaries follow from \cite{lmn1993ac0}. However, much stronger $\ACZ$ lower bounds for $\MAJORITY$ hold and are obtained from other techniques \cite{razborov1987lower, smolensky1987algebraic}. We demonstrate next that these stronger bounds and techniques do not apply to $\QACZ$. 
\subsubsection{Average case quantum advantage.} 
We show that $\QACZ$ circuits can achieve $(1-1/\polylog(n))$-correlation with $\MAJORITY$ on average. Classically,  any randomized $\ACZ$ circuit can be de-randomized while preserving its average case correlation \cite{arora2009computational}. However, randomness in quantum algorithms cannot be distilled in this manner and crucially, for inputs $\x$ whose hamming weight is far from $n/2$, the output of our circuit is not deterministic.

\cite{grier2026tc0} showed that the state $\ket{W_n}$ can be constructed in $\QACZ$. We observe that any superposition of the states $\ket{0^n}$ and $\ket{W_n}$ can be constructed in $\ACZ$. Furthermore, we can implement the map $\ket{0} \mapsto \ket{0^n}$ and $\ket{1} \mapsto \ket{W_n}$ in $\QACZ$ (\cref{lem:ctrlW}). 
This allows us to implement the following map in $\QACZ$, which serves as a ``weak fanout" gadget. 

\begin{restatable*}[Poor man's $n$-fanout in $\QACZ$]{corollary}{poormansfanoutmap}\label{lem:poormanfanout}
There exists a $\QACZ$ circuit acting on an input qubit 
$x$, that cleanly achieves the following mapping,
$$\ket{0}_x \ket{0^n} \mapsto \ket{0}_x \ket{0^n}, \quad \ket{1}_x \ket{0^n} \mapsto \ket{1}_x \ket{1/n}^{\tens n}.$$
where $\ket{\eps} := \sqrt{\eps}\ket{1} + \sqrt{1-\eps}\ket{0}$ for $\eps \in [0,1]$.
\end{restatable*}
\noindent Implementing Poor man's $n^{1.1}$-fanout in $\QACZ$
would imply $\PARITY \in \QACZ$ because the state $\frac{1}{\sqrt{2}} \ket{0} \ket{0^{n^{1.1}}} + \frac{1}{\sqrt{2}} \ket{1} \ket{1/n}^{n^{1.1}}$ can be made to have high felinity on $\approx n^{0.1}$ targets after a change of basis, similar to our Dicke states reduction. Poor man's $n$-fanout is sufficient for our purposes to amplify correlations with threshold functions. 
We note that \cite{grier2024threshold} were the first to show the existence of a gadget with these properties in $\QACZ$ (refer to their Appendix C). They do so by mapping an input $\ket{\x}$ to the state $\ket{W_\x}$ defined as, 
\begin{align}
    \ket{W_{\x}} := \sum_{i \in [n]} (-1)^{x_i} \ket{0^{i-1} 1 0^{n-i}}.
\end{align}
We remark that a similar construction is likely possible using their gadget.

A uniformly random input has its hamming weight reasonably concentrated about $n/2$. Thus, at a high level our strategy is to approximate $\MAJORITY$ well on these strings. On input $\x$, we can apply the Poor man's $n$-fanout gadget to each qubit to obtain $n$ independent weak copies of the input, each of the form $\ket{x_1/n} \ket{x_2/n} \dots \ket{x_n/n}$. We show that the following test can be performed in $\QACZ$ on a weak copy of this form.
\begin{restatable*}[Weak copy test]{lemma}{oneweakcopytest}\label{lem:onecopytest}
For any threshold $t \in [n]$, there exists a $\QACZ$ circuit $C$  such that, for every $\x \in \bin^n$ with $\big||\x|-t\big|=O(\sqrt n\,\polylog(n))$, the state $\rho(\x)$ on the designated register ``$\out$" of $C \ket{x_1/n}\cdots\ket{x_n/n} \ket{0}_{\out}$ satisfies, 
$$\braket{0 | \rho(\x) | 0}= 1 - \lambda_t(1+o(1))\left(\frac{|\x|-t}{n}\right)^2$$
where $\lambda_t = \frac{e^{-t/n}}{1+t^2/n^2} = \Theta(1)$.
\end{restatable*}
\noindent Then, amplifying using the $n$ weak copies and $a = \polylog(n)$ fanout, we obtain a circuit that outputs $1$ with probability at least $1-1/n$ if $\big||\x|-t\big|<\frac{\sqrt n}{2a}$, and outputs $0$ with probability at least $1-1/n$ if $\big||\x|-t\big|>\frac{\sqrt n}{a}$ (\cref{lem:approxtcircuit}). Finally, we perform this test for $\polylog(n)$ values of $t$ spaced out to cover the interval concentrated about $n/2 \pm \sqrt{n}/\polylog(n)$. Choosing these values of $t$ carefully gives the desired correlation with $\MAJORITY$ stated below.
\begin{restatable*}[$(1-1/\polylog(n))$-correlation with Majority in $\QACZ$]{theorem}{ezmajority}\label{thm:ezmajority}
For any constant $c$, there exists a $d:=d(c)$ such that there is a family of $\QAC$ circuits of depth $d$ and size $\poly_c(n)$ that achieve $(1-1/\log^c(n))$ correlation with $\MAJORITY$.
\end{restatable*}
\noindent Combining with existing lower bounds of \cite{razborov1987lower, smolensky1987algebraic} for $\ACZ$ gives the following separation.  
\begin{corollary}[Average-case constant-depth quantum advantage]
There exist Boolean functions $f:\bin^n \to \{-1,1\}$ that have $o(1)$ correlation with any $\ACZ$ circuit, yet have $1 - o(1)$ correlation with some $\QACZ$ circuit.
\end{corollary}
\section{Preliminaries}

A depth-$d$ $\qac$ circuit on $n$ input qubits and $m$ ancilla qubits initialized to $\ket{0^m}$ consists of $d$ layers of multi-qubit gates, in which each qubit participates in at most one gate, and arbitrary single-qubit unitaries taken to be free. We adopt the formulation of \cite{rosenthal2021qac0}, where each multi-qubit gate is a reflection about a product state. Formally, a gate $G(S)$ on subset $S$ of qubits has the form $G(S) = (I - 2\kb{\vth}_S)$ for some product state $\ket{\vth}_S$. The single qubit unitaries are arbitrary and do not contribute to the depth. This is equivalent to \cite{moore1999qac0}'s original definition where each multi-qubit gate is a Toffoli gate \cite{rosenthal2021qac0}.

In most prior works on $\QACZ$, the output of a circuit on input $\x$ is given by the probability of measuring $\ket{1}$ on a designed register $t$. We consider an equivalent observables-based definition, which is more natural in the context of boolean analysis. 
\begin{definition}[Single output circuit]\label{def:booleanckt}
Let $C$ be a $n$-input $\QACZ$ circuit acting on $m$ ancillae, with a designated target register $t$. The output of $C$ on input $\x \in \bin^n$ is a randomized variable $\in \bpm$, given by the outcome of measuring the observable $Z_t \tens I$ on $C \ket{\x} \ket{0^m}$. The associated real-valued boolean \emph{function computed by $C$}, $f_C : \bin^n \to \bpm$,  describes the expected output on input $\x$,  
$$f_C(\x) := \braket{\x, 0^m | C^\dag \cdot Z_t \tens I \cdot C | \x, 0^m}$$
\end{definition}

We also consider $\QACZ$ state preparation circuits that do not act on any input qubits---instead starting with $\ket{0^m}$ on all registers. In this case, the parameter $n$ refers to the number of qubits in the target state and $m$ is required to be $\poly(n)$.
\begin{definition}[State preparation circuit]
 A circuit $C$ prepares an $n$-qubit state $\ket{\psi}_X$ using $m$ qubits if, 
 $$C \ket{0^m} = \ket{\psi}_X \ket{\varphi}$$
 for some state $\ket{\varphi}$, unentangled with $\ket{\psi}_X$ on the remaining qubits.  Furthermore $C$ \emph{cleanly} prepares $\ket{\psi}_X$ if $\ket{\varphi} = \ket{0^{m-n}}$.
\end{definition}
\noindent All our state construction upper bounds are clean preparations. When proving reductions, we don't enforce clean preparations. 

\subsection{Notation}\label{sec:gatenotation}
We use contractions such as $\bra{\nu}_T$ on states acting on larger sub-systems, $\ket{\psi}_{T,A}$. The notation $\bra{\nu}_T \cdot \ket{\psi}_{T,A}$ refers to the un-normalized vector obtained by tracing out the $T$ sub-system in $\kb{\nu}_T \tens I_A \cdot \ket{\psi}$. We also use the shorthand $\bra{\psi} \cdot M \cdot \ket{\psi}$ for the operator given by $(\bra{\psi} \tens I) M (\ket{\psi} \tens I)$ when it is clear that the matrix $M$ acts on space larger than $\ket{\psi}$. Then, \cref{def:booleanckt} can be rewritten, $f_C(\x) := \braket{\x,0^m | C^\dag Z_t C| \x,0^m}$

For $0 \leq \eps \leq 1$, we define $\ket{\eps}$ to be the single-qubit state given by $\ket{\eps} = \sqrt{\eps} \ket{1} + \sqrt{1-\eps} \ket{0}$. This is consistent with the usual definitions of $\ket{0}, \ket{1}$. For $0 \leq \gamma \leq 1$, define $\Rot_{\gamma}$ to be the single-qubit unitary
\begin{align}
    \Rot_\gamma := \begin{bmatrix}
        \sqrt{\gamma} & \sqrt{1-\gamma} \\ 
        \sqrt{1-\gamma} & -\sqrt{\gamma}
    \end{bmatrix}.
\end{align}

\subsubsection{\safeqaczf-completeness} $\QACZF$ can be equivalently defined as the quantum analog of $\TCZ$, consisting of $\QACZ$ and arbitrary \THRESHOLD$^n_k$ gates, i.e, $\QACZF= \QTCZ$  \cite{hoyer2005fanout, grier2024threshold, grier2026tc0}.  The notion of $\QACZF$-completeness can be viewed as the quantum analog of $\TCZ$-completeness \cite{agrawal2001reducing}. Augmenting $\ACZ$ gates with a gate for any $\TCZ$-complete problem provides the ability to implement any other problem in $\TCZ$. Canonical $\TCZ$-complete problems (under $\ACZ$ reductions) consist of sorting, integer division and iterated multiplication \cite{chandra1984constant, hesse2002uniform}, but notably do not include $\PARITY_n$ ($\ACZ$$[2] \subsetneq$ $\TCZ$) \cite{razborov1987lower, smolensky1987algebraic}. On the other hand, $\QACZF$-complete problems (under $\QACZ$ reductions) include $\PARITY$, $\MAJORITY$ and state synthesis of any nekomata (uniform mixture of $\ket{0^n}$ and $\ket{1^n}$) \cite{moore1999qac0, rosenthal2021qac0, grier2024threshold}. 

\subsection{Analysis of boolean functions}
In boolean analysis it is convenient to have $f(\x) \in \bpm$.
To accommodate standard basis inputs of the form $\ket{\x}$ where $\x \in \bin^n$, we define boolean functions such as $\PARITY$, $\MAJORITY$ as $\bin^n \to \bpm$ instead of $\bpm^n \to \bpm$.
For a string $\x \in \bin^n$, we use $\vlr{\x}$ to denote the hamming weight of $\x$. Then, $\PARITY_n(\x) := (-1)^{|\x|}$ and we use the following convention for $\MAJORITY$, 
\[
\MAJORITY_n(\x) :=
\begin{cases}
1 & \text{if } |\x| \leq  n/2, \\
-1 & \text{if } |\x| > n/2.
\end{cases}
\]
A powerful technique for analyzing boolean functions is their Fourier decomposition, for a reference we point the reader to \cite{o2014analysis}. The characters $\chi_S : \bin^n \to \mathbb{R}$ defined by $\chi_S(x) := (-1)^{\oplus_{i \in S} x_i}$ form an orthonormal basis for the vector space of real-valued boolean functions. Projecting $f$ onto this basis reveals the Fourier decomposition of $f$,
\begin{equation}
    f(x) = \sum_{S \subseteq [n]} \hat{f}(S)\chi_S(x).
\end{equation}
By Parseval's, $\sum_{S \subseteq [n]} \hat{f}(S)^2 = 1$. The correlation of real valued boolean functions is defined by:

\begin{definition}[Correlation of real valued boolean functions ] 
   The correlation of two boolean functions $f,g :  \bin^n \to [-1,1]$ is defined as, 
    $$\corr(f,g) := \Ex_{\x \sim \bin^n} f(\x) g(\x).$$
\end{definition}

\noindent By Parseval's again, $\corr(f,g) = \sum_{S \subseteq [n]} \wh{f}(S) \cdot \wh{g}(S)$.
Recall the real-valued boolean function  $f_C : \bin^n \to [-1,1]$ associated with the output of $C$,  from  \cref{def:booleanckt}.  

\begin{definition}[$\QACZ$ circuit correlation with boolean function]\label{def:cktcorr}
Given a circuit $C$ on $n$ input qubits and $m$ ancilla $A$, with a designated output register $t$, the correlation of $C$ with a boolean function $g(\x) : \bin^n \to \bpm$,  $\corr(C,g)$ is defined to the correlation between $f_C$ and $g$, given by, 
\begin{align*}
\corr(f_C,g) &:= \E_{\x \sim \bin^n} f_C(\x) g(\x)\\
&= \E_{\x \sim \bin^n} \braket{\x,0^m|C^\dag Z_t C| \x,0^m} g(\x)
\end{align*}
\end{definition}



\subsection{Procedures in \safeqacz}

\noindent Note the following form of the standard un-computing ancillae technique for quantum circuits due to \cite{bennett1973logical}. We include a proof of this specific formulation in \cref{sec:addproofs}. 
\begin{restatable}[Approximate clean computation \cite{bennett1973logical}]{claim}{approxcleancomp} \label{fact:cleancomp}
  Let $C$ be a $n$-input depth $d$ single-output $\QAC$ circuit using $m$ ancillae such that, on input $\x \in \bin^n$, the output register measures to $\ket{0}$ with \emph{probability} $p_0(\x)$ and $1$ with probability $p_1(\x) = 1-p_0(\x)$. 
  Then there exists a $n$-input $\QAC$ circuit $C'$ of depth $O(d)$ with $m' = m+n+1$ ancillae such that for all $\x$, conditioning on the non-output registers being clean produces the (unnormalized) vector,  
  $$\bra{0^{m'-1}}_A \cdot C'(\x) \ket{\x} \ket{0^{m'}}_{A,t} = p_0(\x) \ket{\x} \ket{0}_t + p_1(\x) \ket{\x} \ket{1}_t$$
 In other words, $C' \cdot \ket{\x} \ket{0^{m'}}$ has \emph{amplitude} exactly $p_b(\x)$ on $\ket{\x} \ket{b} \ket{0^{m'-1}}$, and all the non-output registers are returned to their starting state with probability $p_0(\x)^2 + p_1(\x)^2$. 
\end{restatable}
\noindent We will also need the following primitives.
\begin{fact}\label{fact:reflstate}
    Let $\R_{\ket{\psi}} := I - 2 \ketbra{\psi}{\psi}$.
    If state $\ket{\psi} \ket{0^n}$ can be constructed by a $\QACZ(d, s)$ circuit acting on $\ket{0^m}$, then $\R_{\ket{\psi}}$ can be implemented by a depth $2d+1$ $\QACZ(2d+1,2s+1)$ circuit with the same number of ancillae. 
\end{fact}

\begin{fact}[Poly-logarithmic Fanout in $\QACZ$ \cite{rosenthal2021qac0, grier2026tc0}]\label{fact:polylogfan}
For any $k = \polylog(n)$, \[\FANOUT_k \in \QAC(\Theta(1),\poly(n)).\]
\end{fact}


%
\noindent  \cite{grier2024threshold} also showed that the exact amplitude amplification technique of \cite{Grover, Brassard2000QuantumAA} can be performed in $\QAC$. We will use the following sub-routines that are equivalent in the single-round special case.  
\begin{fact} [Claim 2.4 of \cite{joshi2026improvedlowerboundsqac0}]\label{fact:constnekobranch}
Suppose there exists a depth-$d$ $m$-qubit $\QACZ$ circuit to construct $\ket{\psi}_{T,A} = C \ket{0^m}$ satisfying, 
$$\Vlr{\bra{0^n}_T \cdot \ket{\psi}}^2 \geq 1/4, \quad \text{and} \quad \Vlr{\bra{1^n}_T \cdot \ket{\psi}}^2 \geq 1/4.$$
Then, there is a circuit $C'$ of depth $d' \leq 3(d+2)$ using $m' = m+2$ ancillae that exactly synthesizes an $n$-nekomata.
\end{fact}

\begin{claim}\label{cl:hiampamp}
Let $\ket{\psi}$ be any $n$ qubit state. Suppose there exists a depth-$d$ $m$-ancilla $\QAC$ circuit $C$ that cleanly prepares the following $n+1$-qubit state $\ket{\varphi}$,
$$\ket{\varphi}_{T,t} = \sqrt{\alpha} \ket{\psi}_T \ket{1}_t + \sqrt{1-\alpha} \ket{\bad}_T \ket{0}_t,$$ for some $\alpha \geq 1/5$.
Then, there exists a depth-$3d+2$ $\QAC$ circuit $C'$ that cleanly constructs $\ket{\psi}$, i.e, $C' \ket{0^{m+n+2}} = \ket{\psi} \ket{0^{m+2}}$.
\end{claim}
\begin{proof}
Let $\gamma := 1/(2\alpha)$, and let $\ctr{\R_{\gamma}}$ denote the controlled $\R_\gamma$ gate.
Using a fresh ancilla $a$, we can prepare the following state cleanly in depth $d+1$, 
\begin{align}
    \ket{\psi_0} &:= \ctr{\R_{\gamma}}(t,a) \cdot \ket{\varphi}_{T,t} \ket{0}_a \\
&=  \sqrt{\alpha} \ket{\psi}_{T} \ket{1}_a \R_{\gamma} \ket{0}_t + \sqrt{1-\alpha} \ket{\bad}_T \ket{0}_a \ket{0}_t \\
    &= 1/2 \ket{\psi}_{T} \ket{1}_a \ket{1}_t + \sqrt{3}/2 \ket{\bad'}_{T,t} \ket{0}_a
\end{align}
Applying a $Z$ gate on $a$, we obtain,
\begin{align}
    \ket{\psi_1} &:= - 1/2 \ket{\psi}_{T} \ket{1}_a \ket{1}_t + \sqrt{3}/2 \ket{\bad'}_{T,t} \ket{0}_a.
\end{align}
Then, applying $\R_{\ket{\psi_0}}$ which can be implemented in depth $2d+1$ by \cref{fact:reflstate}, we obtain
\begin{align}
    \ket{\psi_2} &:= (I - 2\kb{\psi_0}) \ket{\psi_1} \\
    &= \ket{\psi}\ket{1}_{t,a}.
\end{align}
Finally to clean up, we can apply $X$ gates to $t,a$ without increasing depth. 
\end{proof}

\section{Fourier spectrum of \safeqacz}\label{sec:qlmn}
We will first show a key connection between the Fourier mass at levels $\geq k$ of a circuit and synthesis of nekomata-like states.
\threshbranchwt
\begin{proof}
Let $C'$ be the approximate clean computation circuit obtained from $C$ through \cref{fact:cleancomp}, such that under the following register labeling, for any $\x \in \bin^n$, 
\begin{align}
 \bra{0^{m' - 1}}_A \cdot C' \ket{\x}_X \ket{0^{m'}}_{A,t} &= \ket{\x}_X \lr{p_0(\x) \ket{0} + p_1(\x) \ket{1}}_t.
\end{align}
Note this vector is not necessarily a normalized state. Now consider the resulting state on input $\ket{+^n}$,
\begin{align}\label{eq:cleaninp}
    \ket{\psi} := C' \ket{+^n}_X \ket{0^{a+1}}_{A,t}
\end{align}
We will show $\ket{\psi^*} = H^{\tens}_X \ket{\psi}$ has the required amplitudes on the $X$ register. 
\paragraph{Required amplitude on $\ket{0^{n}}$ branch.}
From \cref{eq:cleaninp},
\begin{align}
\Vlr{\bra{+^n}_X \cdot \ket{\psi}} &\geq \Vlr{\bra{+^{n+1}}_{X,t} \cdot \ket{\psi}} \\
&\geq \bra{+^{n+1}}_{X,t} \bra{0^{m'-1}}_A \cdot \ket{\psi} \\
&= \bra{+^{n+1}}_{X,t} \cdot \lr{ \frac{1}{\sqrt{2^n}} \sum_{\x \in \bin^{n}} p_0(\x) \ket{\x}_X \ket{0}_t + p_1(\x) \ket{\x}_X \ket{1}_t} \\ 
&= \frac{1}{2^n} \sum_{\x,\y \in \bin^{n}} \braket{\y|\x} \cdot \bra{+}_t \cdot \lr{p_0(\x)  \ket{0} + p_1(\x) \ket{1}}_t \\ 
&= \Ex_{\x \sim \bin^n} \blr{ \frac{p_0(\x) + p_1(\x)}{\sqrt{2}}} \qquad \qquad \qquad \qquad \text{(use $\braket{\x|\y} = \delta_{\x,\y}$)}\\ 
&= \frac{1}{\sqrt{2}}.
\end{align}
Thus $\Vlr{\bra{0^n} \cdot \ket{\psi^*}}^2 \geq \frac{1}{2}$.
\paragraph{Required amplitude on threshold branch.}
For any $+1$ eigenstate $\ket{\T_k}$ of $\Pi^n_{\geq k}$, 
\begin{align}\label{eq:eigbnd}
\Vlr{\pnk \cdot \ket{\psi^*}} \geq \Vlr{\bra{\T_k} \cdot \ket{\psi^*}}.
\end{align}
We will construct a $\ket{\T_k}$ that serves as a gadget to extract the Fourier mass. Re-label the standard basis vectors $\ket{\y}$ for $\y \in \bin^n$ to $\ket{S}$ for  $S \subseteq [n]$, by mapping each string to its corresponding subset. Consider the following eigenstate, 
\begin{align}
    \ket{\T_k} &= \frac{1}{\sqrt{\gamma_k}} \sum_{\substack{S \subseteq [n] \\ |S| \geq k}} \wh{f_C}(S) \ket{S} 
\end{align}
where, $\gamma_k = \wgk[f_C]$, ensuring $\ket{\T_k}$ is normalized. In the Hadamard basis, 
\begin{align}
\ket{\wh{T_k}} &:=  H^{\tens n} \ket{\T_k}  \\
&= \frac{1}{\sqrt{\gam_k}}  \sum_{\substack{S \subseteq [n] \\ |S| \geq k}} \wh{f_C}(S) \ket{\vec{-}}_S \ket{\vec{+}}_{\ov{S}}.
\end{align}
Using \cref{eq:cleaninp},
\begin{align}
\Vlr{\bra{\T_k} \cdot \ket{\psi}} &\geq \vlr{\bra{\wh{T_k}}_T \bra{-}_t \bra{0^{m'-1}}_A \cdot \ket{\psi}} \\
&\geq \bra{\wh{T_k}}_T \bra{-}_t \bra{0^{m'-1}}_A \cdot \ket{\psi} \\ 
&= \bra{\wh{T_k}}_T \bra{-}_t \cdot \lr{  \frac{1}{\sqrt{2^n}} \sum_{\x \in \bin^n} \ket{\x} \lr{p_0(\x) \ket{0} + p_1(\x) \ket{1}}} \\
&= \frac{1}{\sqrt{2^n}} \sum_{\x \in \bin^n} \braket{\wh{\T_k}| \x} \cdot \underbrace{\lr{p_0(\x) - p_1(\x)}}_{f_C(\x)/\sqrt{2}} \\ 
&= \frac{1}{\sqrt{2 \cdot 2^n \cdot \gam_k}} \sum_{\substack{\x \in \bin^n \\  S \subseteq [n], |S| \geq k}} \wh{f_C}(S) \cdot \underbrace{\lr{\bra{\vec{-}}_S \bra{\vec{+}}_{\ov{S}}} \cdot \ket{\x}}_{2^{-n/2} \cdot \chi_S(\x)} \cdot f_C(\x) \\ 
&= \frac{1}{\sqrt{2\gam_k}} \sum_{\substack{S \subseteq [n] \\ |S| \geq k}} \wh{f_C}(S) \cdot \Ex_{\x \sim \bin^n} \chi_{S}(\x) \cdot f_C(\x) \\ 
&= \frac{1}{\sqrt{2\gam_k}} \sum_{\substack{S \subseteq [n] \\ |S| \geq k}} \wh{f_C}(S)^2 \\
&=  \sqrt{\wgk[f_C]/2}.
\end{align}
Therefore, from \cref{eq:eigbnd}, $\Vlr{\pnk \cdot \ket{\psi^*}}^2 \geq \wgk[f_C]/2$.
\end{proof}

To prove the equivalence between \cref{conj:qlmn} and $\PARITY_n$, we will first reduce $\PARITY_n$ to the following weak-nekomata state which we show is $\QACZF$-complete. 

\begin{lemma}[Skewed nekomata to exact parity]\label{lem:skewneko_to_exact}
Let $C$ be a depth-$d$ $\QACZ$ circuit acting on $m$ qubits that constructs the state $\ket{\psi}_{T,A} = C \ket{0^m}$ satisfying $\Vlr{\kb{0^n}_T \cdot \ket{\psi}}^2 \geq 1/4$ and $\Vlr{\kb{1^n}_T \cdot \ket{\psi}}^2 \geq \eps$ for some non-negligible $\eps$.  
Then, there exists a depth $d' = O(d)$ circuit that, using $m' = m \cdot \poly(n)$ ancillae prepares an exact nekomata state.  
\end{lemma}
\begin{proof}
WLOG, for a real and non-negative $\gamma \geq 1/2$ the state $\ket{\psi}$ can be written as 
\begin{align}
   \ket{\psi} &= \sqrt{\gamma} \ket{0^n}_T \ket{\alpha_0}_A + \sqrt{\eps} \ket{1^n}_T \ket{\alpha_1}_A + \sqrt{1-\gamma - \eps} \ket{\perp}_{T,A}.
\end{align}
Now, using a fresh ancilla $a$ and two gates controlled on $\ket{0^n}$ and $\ket{1^n}$ to mark the nekomata branches, we obtain
\begin{align}
    \ket{\psi_1} &= (I - 2\kb{0^n} \tens \kb{+}_a) \cdot (I - 2\kb{1^n} \tens \kb{+}_a) \cdot \ket{\psi} \ket{0}_a\\
    &= \sqrt{\gamma} \ket{0^n}_T \ket{\alpha_0}_A \ket{1}_a + \sqrt{\eps} \ket{1^n}_T \ket{\alpha_1}_A \ket{1}_a +  \sqrt{1-\gamma - \eps}  \ket{\perp}_{T,A} \ket{0}_a.
\end{align}
Then we apply \cref{cl:hiampamp} to amplify the ``good'' amplitude and discard the cleaned up $a$ register to produce the state
\begin{align}
    \ket{\psi_1} &= \sqrt{1-\eps_1} \ket{0^n}_T \ket{\alpha_0}_A + \sqrt{\eps_1} \ket{1^n}_T \ket{\alpha_1}_A.
\end{align}
where $\eps_1 = \eps / \gamma \geq \eps$ is still non-negligible. 

Next observe that the parallel amplification technique of \cite{rosenthal2021qac0} can now be used to obtain a high fidelity nekomata using $m_1 := \Theta(1/\eps_1^2)$ copies and one layer.
Imagine arranging these copies into a grid, each copy in a row. Then for each column $i \in [n]$, apply an OR of all qubits with output unto a fresh register $q_i$. The registers $Q = \clr{q_1 \dots q_n}$ now make up the new nekomata targets. To see this, the probability of this measuring to all $0$ is
\begin{align}
   \Vlr{\bra{0^n}_Q \cdot \ket{\psi}}^2 &= \lr{\sqrt{1-\eps_1}}^{2m_1} \approx 1/e.
\end{align}
The remaining probability is all $1$s. 
This produces a high-fidelity approximation of a nekomata, allowing us to apply \cref{fact:constnekobranch} to obtain an exact nekomata and thus $\FANOUT_n$ or $\PARITY_n$. 
\end{proof}

We now show that a generalized ``threshold-block'' nekomata is $\QACZF$-complete.

\twobranchdicke
\begin{proof}
Let $D$ be the distribution over $\{0,1\}^n$ supported on strings of Hamming weight at least $k$ given by the squared amplitudes of $\frac{\pnk \ket{\psi}}{\Vlr{\pnk \cdot \ket{\psi}}}$ in standard basis.
Then we will show there exist $\ell = \Theta\!\left(\frac{k}{\log^2 k}\right)$ disjoint nonempty subsets $S_1,\dots,S_\ell \subseteq [n]$ such that
\begin{align}
\Pr_{\x \sim D}\!\left[\forall j \in [\ell],\; \x|_{S_j} \not\equiv 0 \right] \ge c
\end{align}
for some absolute constant $c>0$.
Let $\ell = \left\lfloor \frac{k}{\log^2 k} \right\rfloor.$ Partition $[n]$ uniformly at random into $\ell$ nonempty disjoint sets $S_1,\dots,S_\ell$ of (almost) equal size. Then for every $j$,
\begin{align}
|S_j| \ge \frac{n}{2\ell}.
\end{align}
Note the standard bound for any $t \geq k$,
\[
\frac{\binom{n-|S_j|}{t}}{\binom{n}{t}} \le \left(1-\frac{|S_j|}{n}\right)^t  
\le e^{-t|S_j|/n} \le e^{-k/(2\ell)}
=  e^{-\Theta(\log^2 k)}
.\]
Fix any $x \in \{0,1\}^n$ with Hamming weight $t \ge k$, and let $T = \mathrm{supp}(x)$. For a fixed $j$,
\begin{align}
\Pr[T \cap S_j = \varnothing] &= \frac{\binom{n-|S_j|}{t}}{\binom{n}{t}} \leq e^{-\Theta(\log^2 k)}.
\end{align}
By a union bound over all $\ell$ blocks,
\begin{align}
\Pr\big[\exists j : T \cap S_j = \varnothing\big] \le \ell \cdot e^{-\Theta(\log^2 k)}.
\end{align}
Since $\ell \le k$, the right-hand side is $o(1)$, and in particular for large enough $k$,
\begin{equation}
\Pr\big[\forall j,\; T \cap S_j \neq \varnothing\big]
\ge \frac{9}{10}.
\end{equation}
Thus, for a random partition, every fixed $x$ of weight at least $k$ hits all $\ell$ blocks with $\Omega(1)$ probability. Averaging over $x \sim D$ implies that there exists a fixed choice of $S_1,\dots,S_\ell$ achieving the desired bound.

Apply an $\OR$ gate for each of these subsets onto a fresh ancilla. Then, \cref{lem:skewneko_to_exact} completes the reduction.
\end{proof}
Using this reduction, we can now prove the equivalence between $\PARITY$ and $\text{QLMN}$.

\subsection{Parity equivalence of Quantum LMN}

\begin{corollary}[High Degree Fourier mass to $\PARITY$ computation]\label{cor:qlmntoparity}
Suppose there exists a depth $d$, $m$-ancilla single-output $\QACZ$ circuit $C$ acting on $n$ inputs, such that for some $k \geq n^\delta$,  
$$\wgk[f_{C}] \geq \eps,$$
then, there exists a depth $O(d)$, $O((m+n)\cdot \log^2(k)/\eps)$-ancilla $\QACZ$ circuit $C'$ for $\PARITY_k$ and a depth $O(d/\delta), \tilde{O}((m+n) \cdot n^{1 - \delta}/\eps)$-ancilla $\QACZ$ circuit for $\PARITY_n$.
\end{corollary}

\begin{proof}
 By \cref{lem:qlmntothres} we can construct a $\QACZ$ circuit $C'$ that prepares a state $\ket{\psi}$ satisfying
$\Vlr{ \bra{0^n} \cdot \ket{\psi}}^2 \geq 1/2$ and $\Vlr{\pnk \cdot \ket{\psi}}^2 \geq \eps/2$. 
Therefore, if $\eps$ is non-negligible, \cref{lem:two_branch_dicke} implies a $\QACZ$ circuit for for $\PARITY_{k/\log^2 k}$ with $O(d)$ depth, and $O((m+n)/\eps)$ ancilla. By the classical divide and conquer argument, there exists a $O(d)$ depth, $O((m+n)\log^2(n)/\eps)$ ancilla $\QACZ$ circuit for $\PARITY_k$, and a circuit for $\PARITY_n$ with a $O(1/\delta)$ multiplicative blowup in depth and $n^{1-\delta}$ blowup in ancillae.
\end{proof}

We now have all the pieces to complete the proof of \cref{thm:qlmnequivparity}.
\qlmnequivparity*
\begin{proof}
    Let $S(d,n)$ denote the maximum size such that no depth $d$ $\QAC$ circuit of this size can compute $\PARITY_n$.

    We first prove the RHS implies the LHS. Let $C$ be any single-output, $n$-input $\QAC(d,s)$ circuit. 
    By Corollary~\ref{cor:qlmntoparity}, there exists a $\QAC\left(\Theta(d), \Theta\left(\frac{s \cdot \log^2(k)} {W^{\geq k}[f_\mathcal{C}]}\right)\right)$ circuit for computing $\PARITY_k$ (as the ancilla blowup is proportional to the size blowup). Therefore $\frac{s}{W^{\geq k}[f_\mathcal{C}]} > \tilde{\Theta}(S(\Theta(d), k))$ (as clearly $S(d,k) \geq k$), and hence $s \cdot \tilde{\Theta}\left(S(\Theta(d),k)\right)^{-1} > W^{\geq k}[f_\mathcal{C}]$, which implies $\text{QLMN}\left(\tilde{\Theta}\left(S( \Theta(d), k)\right)\right)$. Note that for any $\gamma$ such that the RHS holds $S \geq \gamma$, and so does $\text{QLMN}\left(\tilde{\Theta}\left(\gamma(\Theta(d),k)\right)\right)$.

    The other direction is simpler. Suppose $\text{QLMN}(\gamma)$ holds. Then any $\QAC(d,s)$ circuit for computing $\PARITY_k$ must have $s > \gamma(d,k)$ and so $\PARITY_k \not \in \QAC(d, \gamma(d,k))$. Appropriate reparameterizations give the stated form.
\end{proof}


\subsubsection{Consequences for \safeqaczf-complete problems}

We now present some consequences of the equivalence of \text{QLMN} and $\PARITY \not \in \QACZ$ for $\QACZF$-complete problems. The first of which is the uniquely quantum phenomenon that any non-negligible correlation with $\PARITY_n$ can be boosted to exact-$\PARITY_n$ in $\QACZ$.

\approxtoexactpar
\begin{proof}
The correlation of $f_C$ with $\PARITY_n$ is equivalent to $\W^n[f_C]$. Therefore the statement follows from \cref{cor:qlmntoparity}.
\end{proof}

We also show that any circuit that achieves high correlation with $\MAJORITY$ suffices to compute $\FANOUT$. 

\hardmajority
\begin{proof}
Assume WLOG that $n$ is odd, otherwise obtain a similar correlation for $\MAJ_{n-1}$ by using an extra ancilla.
It is a known fact \cite{o2014analysis}, that there exists a constant $\alpha$ for any $k$, $\wgk[\MAJ_n] \geq \alpha k^{-1/2}$.
  
For any $k$, let $\MAJ^{\ge k}_n$ be the real-valued Boolean function obtained from $\MAJORITY_n$ by setting all Fourier coefficients of degree at most $k$ to zero. Formally, 
\begin{align}
    \MAJ^{\ge k}_n(x) := \sum_{\substack{S \subseteq [n] \\ |S| \geq k}} \widehat{\MAJ_n}(S)\,\chi_S(x).
\end{align}
Similarly define $\MAJ^{< k}_n$. Then, it holds that
\begin{align}
    \corr(f_C, \MAJ_n) &= \corr(f_C, \MAJ^{\geq k}_n) + \corr(f_C, \MAJ^{< k}_n) \\
    &\leq \W^{<k}[f_C]^{1/2}\cdot \W^{< k}[\MAJ_n]^{1/2} \\
    & \qquad + \W^{\geq k}[f_C]^{1/2}\cdot \W^{\geq k}[\MAJ_n]^{1/2} \tag{Cauchy-Schwarz} \\
    &\leq \W^{< k}[\MAJ_n]^{1/2} + \W^{\geq k}[f_C]^{1/2} \\
    &\leq 1-\frac{\alpha}{2} k^{-1/2} + \W^{\geq k}[f_C]^{1/2}.
\end{align}

Also, $\corr(f_C, \MAJ_n) \geq 1 - n^{-\delta}$. Hence, for $k' = 0.01\alpha^{2} n^{2\delta}$,   
\begin{align}
   \W^{\geq k'}[f_C]^{1/2} &\geq 0.5 \alpha k'^{-1/2} - n^{-\delta} \\
                 &=  5 \cdot n^{-\delta} - n^{-\delta}\\  &= \Omega(n^{-\delta}).
\end{align}
Thus, $\wg{k'}[f_C]=\Omega(n^{-2\delta})$, and applying \cref{cor:qlmntoparity}, this implies a $O(1/\delta) = O(1)$ depth $\QACZ$ circuit for $\PARITY_n$.
\end{proof}
\noindent This result also implies that a worst-case, approximate $\MAJORITY$ circuit can be made exact (due to $\FANOUT$).

\hardworstmajority
\begin{proof}
For each input, make $k$ copies using $\FANOUT_n$ and run $C$ in parallel. Then, compute $\MAJORITY_k$ of all the runs to output the final answer. From a standard Chernoff bound, $k$ can be chosen to be $\poly(1/\eps) = \polylog(n)$ to achieve at most $1/n$ worst case error, while implementing the $\FANOUT_k, \MAJORITY_k$ in $\QACZ$. The resulting $\QACZ$ circuit has $\geq (1-1/n)$-correlation with $\MAJORITY$, and therefore the rest of the argument follows from \cref{thm:hardmajority}.
\end{proof}

\subsection{Felinity and \safeqaczf-complete states}\label{sec:felinity}
Recall the notion of Felinity, restated below. 
\felinitydef*
\noindent We will first show the two main properties of felinity that will be used in subsequent reductions. These are: (1) felinity behaves naturally under trace distance (\cref{lem:feltd}) (2) a circuit for preparing an $n$-qubit $\rho$ produces a circuit for $\feln(\rho)$ correlation with parity $\PARITY_n$ (\cref{lem:felinity_to_par}).

\begin{lemma}[Felinity is Lipschitz-continuous]\label{lem:feltd}
Let $\rho,\sigma$ be $n$-qubit states. Then
$$|\feln(\rho)-\feln(\sigma)| \le 8\,\TD(\rho,\sigma),$$
where
$\TD(\rho,\sigma) := \frac12 \|\rho-\sigma\|_1.$
\end{lemma}
\begin{proof}
Consider the computational basis measurement on $2n$-qubit states. For any $\y,\z \in \bin^{n}$ and a $n$-qubit state $\rho'$, define, 
\begin{align}
q_{\rho'}(\y\z) &:= \braket{\y,\z | \rho' \tens \rho' |\y,\z} = \braket{\y|\rho'|\y} \braket{\z|\rho' |\z}  
\end{align}
This is precisely the probability of the measurement outcome $\y\z \in \bin^{2n}$ on the $2n$-qubit state $\rho' \tens \rho'$. From triangle inequality and contractivity of trace distance under measurement, 
\begin{align}
  \frac12  \sum_{\y,\z \in \bin^{n}} \vlr{q_{\rho}(\y\z) - q_{\sig}(\y\z)}  &\leq \TD(\rho \tens \rho,\sig \tens \sig) \\
  &\leq 2\TD(\rho, \sig).
\end{align}
For $\y\in\{0,1\}^n$, let $\ov{\y}$ be the string obtained by flipping all bits of $\y$, i.e, $\ket{\ov{\y}} = X^{\tens n}\ket{\y}$. Then, 
\begin{align}
    \vlr{\feln(\rho) - \feln(\sig)} &= 2 \vlr{\sum_{\y \in \bin^n} q_{\rho}(\y\ov{\y}) - q_{\sig}(\y\ov{\y})} \\
    &\leq 2 \sum_{\y \in \bin^n} \vlr{q_{\rho}(\y\ov{\y}) - q_{\sig}(\y\ov{\y})} \\
    &\leq 2 \sum_{\y,\z \in \bin^n} \vlr{q_{\rho}(\y\z) - q_{\sig}(\y\z)} \\
    &\leq 8 \TD(\rho, \sig).
\end{align}
\end{proof}
Recall that we treat the output of the circuit as as a $\pm 1$ valued measurement outcome. Any Hermitian matrix $\O$ is a valid \emph{observable} corresponding to a measurement with possible outcomes among the eigenvalues of $\O$. For a state $\ket{\psi}$, the expected outcome of measuring $\O$ on $\ket{\psi}$ is given by $\braket{\psi | \O |\psi}$. 
For $n$-input, $m$-ancillae $\QACZ$ circuit $C$ with a designated target $t$, define the following observable, 
 \begin{align}
 \O_C := \bra{\0^m}_A \cdot {C^\dag Z_t C} \cdot \ket{\0^m}_A
 \end{align}
Physically, $\O_C$ can be interpreted as the Heisenberg evolution of the output measurement. Mathematically, $\O_C$ is simply a $2^n \times 2^n$ Hermitian matrix with eigenvalues in $[-1, 1]$ satisfying $f_{C}(\x) = \braket{\x | \O_C | \x}$.
With this formalism in mind, we proceed to show the main property of felinity.

\parfromfeln
\begin{proof}
Let $T = \clr{t_1, t_2\dots  t_n}$ and $A$ be the set of qubits (including $T$) of $C$. To construct $C'$, first label the input qubits $X = \clr{x_1, x_2 \dots x_n}$ and define a layer of $\cz$ gates, $L := \bigotimes_{i \in [n]} \cz(x_i,t_i)$.Using a fresh ancilla $a$, designated as the target register, define the following reversible-$\NOR$ gate,
\begin{align}\label{eq:orgate}
    G &:= (I - \kb{0^m})_A \tens I_a + \kb{0^m}_A \tens X_a  \\
    &=  I - 2\kb{0^m}_A \tens \kb{-}_a
\end{align}
This gate flips the target $a$ when the $\NOR$ of all the controls is $1$ and is a valid $\QACZ$ gate. Then, $C'$ is given by, 
\begin{align}
    C' := G C^\dag L C
\end{align}
and is a depth $2d+2$ circuit with $m+1$ ancillae. 

We will now analyze the correlation of $C'$ with $\PARITY(\x)$ as defined in \cref{def:cktcorr}. 
First define the observable $\O_{C'}$ corresponding to measurement outcome of $C'$, s.t $f_{C'} = \braket{\x | \O_{C'} | \x}$
\begin{align}
    \O_{C'} &= \braket{0^{m}_A, 0_a | (C')^\dag Z_a \tens I_A C' | 0^{m}_A, 0_a} \\
    &= \braket{0^{m}_A, 0_a | C^\dag L C G^\dag Z_a G C^\dag L C | 0^{m}_A, 0_a} & \qquad (L = L^\dag)\\ 
    &= \braket{\psi_A, 0_a | L C G^\dag Z_a G C^\dag L  | \psi_A, 0_a} \\ 
    &= \bra{\psi}_A \cdot L C \braket{0_a| G^\dag Z_a G |0_a} C^\dag L^\dag \cdot \ket{\psi}_A \label{eq:here} \\ 
    &= \bra{\psi}_A \cdot L C (2\kb{0^m} - I) C^\dag L^\dag \cdot \ket{\psi}_A \\ 
    &= \bra{\psi}_A \cdot L  (2\kb{\psi} - I)  L^\dag \cdot \ket{\psi}_A \\ 
    &= 2 \bra{\psi}_A \cdot L \cdot \kb{\psi}_A \cdot L \cdot \ket{\psi}_A -  \bra{\psi}_A \cdot L L^\dag \cdot \ket{\psi}_A  \\
    &= \underbrace{2 \bra{\psi}_A   L \kb{\psi}_A L \ket{\psi}_A}_{\O'_{C'}} - I_n
\end{align}
In \cref{eq:here} we used the following derivation for the observable induced on $A,X$ through \cref{eq:orgate}
\begin{align}
    \bra{0_a} G^\dag Z_a G \ket{0_a} &=  \bra{0_a} G^\dag \kb{0}_a G \ket{0_a} -  \bra{0_a} G^\dag \kb{1}_a G \ket{0_a}  \\
    &=  (\kb{0^m}) \cdot (\kb{0^m}) -  (I - \kb{0^m}) (I - \kb{0^m})  \\
    &= \kb{0^m} - (I - \kb{0^m}) \\
    &=  (2\kb{0^m} - I).
\end{align}
Next, we expand $\ket{\psi}$ in the standard basis on $T$. We re-label each basis state $\ket{\y}_T$, for $\y \in \bin^n$ by mapping to a unique subset $S \subseteq [n]$ based on the $1$ bits in $\y$. This associates each basis state with a Fourier character $\chi_S$, giving,
\begin{align}
    \ket{\psi} &= \sum_{S \subseteq [n]} \alpha_{S} \ket{S} \ket{\phi_{S}}_{A \setminus T}.
\end{align}
for some arbitrary states $\ket{\phi_{S}}$ on the remaining ancillae and coefficients $\alpha_S$.  Putting it together, 
\begin{align}
    \eps = \corr({f_{C'},\PARITY}) &= 
     \Ex_{\x \sim \bin^n} \braket{\x | \O_{C'} |\x} (-1)^{|\x|} \\
    &= \Ex_{\x \sim \bin^n} \braket{\x | \O'_{C'} |\x} (-1)^{|\x|} { \  - \Ex_{\x \sim \bin^n} (-1)^{|\x|}} \\
    &= 2 \cdot \Ex_{\x \sim \bin^n} \braket{\psi,\x |  L | \psi} \braket{\psi |  L | \psi, \x} (-1)^{|\x|}  \\
    &= 2 \sum_{S,S' \subseteq [n]} \alpha^*_{S} \alpha_{S'}  \cdot \Ex_{\x \sim \bin^n} \braket{S,\phi_S, \x |  L | \psi} \braket{\psi |  L | S', \phi_{S'}, \x} \chi_{[n]}(\x) \\
    &= 2 \sum_{S,S' \subseteq [n]} \alpha^*_{S} \alpha_{S'}  \cdot \Ex_{\x \sim \bin^n} \chi_S(\x) \chi_{S'}(\x) \braket{S, \phi_{S}| \psi} \braket{\psi | S', \phi_{S'}} \chi_{[n]}(\x) \\
    &= 2 \sum_{S,S' \subseteq [n]} \vlr{\alpha_{S}}^2 \vlr{\alpha_{S'}}^2  \cdot \Ex_{\x \sim \bin^n} \chi_{S\Delta S'}(\x) \chi_{[n]}(\x) \\
    &= 2 \sum_{S \subseteq [n]} \vlr{\alpha_{S}}^2 \vlr{\alpha_{\ov{S}}}^2 .
\end{align}
Observe that $\vlr{\alpha_{S}}^2 = \braket{S|\rho_T |S}$, for $\rho_T = \tr_{A\setminus T} \kb{\psi}$.
Mapping the labels $S$ back to bit-strings gives, 
\begin{align}
    \corr({f_{C'},\PARITY}) &= 2 \sum_{\y \in \bin^n} \braket{\y|\rho_T|\y} \braket{\y | X^{\tens n}  \cdot \rho_T \cdot X^{\tens n} | \y} \\ 
    &= \feln(\rho_T).
\end{align}
\end{proof}
We conclude with the following reduction. 

\felntofanout
\begin{proof}
This follows directly from \cref{lem:felinity_to_par} and \cref{thm:nonnegpar_to_exact}. 
\end{proof}

\subsubsection{Poly weight Dicke states are \safeqaczf-complete}
It is known that any Dicke state $\ket{D^k_n}$ can be constructed in $\QACZF$, i.e using $\FANOUT_n$. We will show that constructing $\ket{D^k_n}$ in constant depth \emph{requires} $\FANOUT_{\min(k,n-k)}$.
\begin{definition}[Uniform $k$-hamming slice / Dicke $\D^n_k$]\label{def:dicke_dist}
The uniform distribution on all $n$-bit strings of hamming weight $k$. The associated state is $\ket{D^n_k}$, the $n$-qubit $k$-weight Dicke state,
    \[\ket{D^n_k} := \binom{n}{k}^{-1/2} \sum_{b \in \{0,1\}^n, |b| = k} \ket{b}.\]
\end{definition}
\noindent We will focus on $k \leq n/2$ WLOG, since $\ket{D^{n-k}_n}$ has the same complexity as $\ket{D^n_k}$.

\begin{fact}[Binomial distribution bounds \cite{grinstead2006grinstead}]\label{fact:binombounds}
  For any sufficiently large $n \in \mathbb{N}$, and any $k \leq n/2$  $$\Pr[\binomd(n,k/n) = k] \geq \frac{e^{-1}}{\sqrt{k}}$$ 
\end{fact}
\noindent To prove $\QACZF$-completeness of Dicke state preparation, we first prove the following reduction from Dicke states to non-negligible felinity states. 
\begin{lemma}\label{lem:dicketofel}
For any $n$ and $k \leq n/2$, there is a single layer $L$ of $\QACZ$ gates acting on $n + \ell$ qubits for $\ell = \lfloor k/\log k \rfloor$, such that, the state 
$\ket{\psi} = L \dkn \ket{0^{\ell}}_T$ satisfies, 
$$\fel{\ell}\lr{\tr_{\ov{T}} \kb{\psi} } \geq \frac{1}{8k}.$$  
\end{lemma}
\begin{proof}
  Partition the qubits into $\ell = \lfloor k/\log k \rfloor$ blocks. Let $B_i$ be the $m_i = |B_i| \geq n/\ell$ qubits in each block, and label a corresponding qubit of $T$ as $t_i$. 
  Let $p = \frac{k}{n}$. Then, $L$ is given by,
\begin{align}
  L &= \bigotimes_{i \in [\ell]} (I_{B_i} - \kb{p}^{\tens m_i} \tens \kb{+}_{t_i}) \\ 
    &= \bigotimes_{i \in [\ell]} \blr{(I_{B_i} - \kb{p}^{\tens m_i}) \tens I_{t_i} + \kb{p}^{\tens m_i} \tens X_{t_i}}.  \label{eq:ascnot}
\end{align}
As described in \cref{eq:ascnot}, each gate can be viewed as a controlled-not gate, that flips $t_i$ whenever $B_i$ lies in the $\kb{p}^{\tens m_i}$ subspace. 
Now, for $\ket{\psi} = L \dkn \ket{0^{\ell}}$, we will first argue that $\Vlr{\bra{1^\ell}_T \cdot \ket{\psi}}^2 \geq \eps$ and $\Vlr{\bra{0^\ell}_T \cdot \ket{\psi}}^2 \geq \eps$. From \cref{eq:ascnot}, we have, 
\begin{align}
  \Vlr{\bra{1^\ell}_T \cdot \ket{\psi}}^2 &= \Vlr{\bra{p}^{\tens n}\cdot \ket{\psi}}^2 \\
                                          &= \Vlr{\bra{p}^{\tens n} \cdot \dkn}^2\\
                                          &= \Pr\blr{\binomd(n,p) = k} \\
                                          &\geq \frac{e^{-1}}{\sqrt{k}}. \qquad \qquad \qquad \qquad \text{(from \cref{fact:binombounds})} \label{eq:branch1}
\end{align}
similarly, letting $\Pi = \bigotimes_{i \in [\ell]} (I - \kb{p}^{\tens m_i})$, 
\begin{align}
  \Vlr{\bra{0^\ell}_T \cdot \ket{\psi}}^2 &= \Vlr{\Pi \cdot \dkn}^2. \label{eq:branch2}
\end{align}
Instead of directly bounding $\Vlr{\Pi \dkn}^2$, we bound $\vlr{\braket{\nu | D^n_k}}^2$ for a vector $\ket{\nu}$ that is \emph{nearly} a $+1$ eigenvector of $\Pi$. Define, $\ket{p_-} = \sqrt{1-p} \ket{0} - \sqrt{p} \ket{1}$.

Observe that, for each $i$,
\begin{align}
  \braket{p_-|p}^{2 m_i} &= (1-2p)^{2m_i} \\
                         &\leq \exp(-4k/\ell)\\
                         &= \frac{1}{k^4}.
\end{align}

Therefore,

\begin{align}
  \Vlr{\Pi \ket{p_-}^{\tens n}}^2 &= \prod_{i \in [\ell]} (1 - \vlr{\braket{p|p_-}}^{2m_i})  \\
                           &\geq \lr{1 - 1/k^4}^{\ell} \\
                           &\geq 1- 1/k^3.
\end{align}
Therefore, $\ket{p_-}^{\tens n}$ is nearly an eigenvector of $\Pi$, and,
\begin{align}
  \Vlr{\Pi \cdot \dkn} &\geq \vlr{\bra{p_-}^{\tens n} \cdot \dkn}  - 1/k^{3/2}  \\
                         &= \Vlr{\bra{p}^{\tens n} \cdot Z^{\tens n} \dkn} - 1/k^{3/2} \\
                         &= \Vlr{(-1)^{k} \cdot \bra{p}^{\tens n} \dkn} - 1/k^{3/2} \\
                         &= \sqrt{\Pr\blr{\binomd(n,p) = k}} - 1/k^{3/2} \\
                         &\geq \frac{e^{-1.01}}{k^{1/4}}.
\end{align}

Hence, $\Vlr{\Pi \cdot \dkn}^2 \gtrsim 1/\sqrt{k}$.

Combining with \cref{eq:branch2} and \cref{eq:branch1}, for $\rho_T = \tr_{\ov{T}} \kb{\psi}$, 
\begin{align}
  \fel{\ell}(\rho_T) &\geq \braket{0^\ell|\rho_T|0^\ell} \braket{1^\ell|\rho_T|1^\ell} \\
                     &\geq \frac{1}{8k}. \qquad \qquad \qquad \qquad \qquad \qquad \text{(using $e^{2.01} \approx 7.5$)} 
\end{align}
\end{proof}

Now we prove the $\QACZF$-completeness of poly-weight Dicke states.

\dicketocat
\begin{proof}
Consider the circuit $C$ from \cref{lem:dicketofel}, and let $\sig^*$, $\sig$ be the state on the $\ell = \lfloor k/\log k \rfloor$-qubit register $T$ after applying $C$ to $\dkn \ket{0^\ell}_T$ and $\rho \tens \kb{0^{\ell}}_T$ respectively.
Then, 
\begin{align}
  \TD(\sig, \sig^*) &\leq \TD(\rho \tens \kb{0^{\ell}}, \dkn \ket{0^{\ell}}) \\
                    &= \TD(\rho, \dkn).
\end{align}
Then, from \cref{lem:feltd}, 
\begin{align}
  \fel{\ell}(\sig) &\geq \frac{1}{8k} - 8\TD(\rho,\dkn) \\
                   &\geq \frac{1}{8k} - \frac{1}{10k} \\
                   &= \Omega(1/k).
\end{align}
By \cref{cor:felnhard}, this implies a depth $O(d)$, $m \cdot \poly(n)$-ancilla circuit for $\FANOUT_\ell$. Note $\log \ell/ \log k = 1-o(1)$. Then, using the recursive nature of $\FANOUT$ \cite{moore1999qac0, rosenthal2021qac0}, this provides a depth $d' = O(d)$, $m' = m \cdot \poly(n)$-ancilla circuit for $\FANOUT_k$, as well as a depth $O(d'/\delta)$, $O(m' n^{1-\delta})$ ancillae circuit for $\FANOUT_n$.
\end{proof}

\polydickeisfanout
\begin{proof}
Either $\rho$ or $X^{\tens n} \rho$ is within $1/k$ trace distance of $\dkn$. Since $\log k/ \log n = \Theta(1)$, \cref{thm:dicketocat} implies a $O(1)$ depth and $\poly(n)$ ancilla circuit for $\FANOUT_n$. 
\end{proof}

\section{Upper bounds for \safeqacz}\label{sec:upper}

We now prove some important building blocks for our construction of a circuit with high $\MAJORITY$ correlation. We will first show that ``poor-man's fanout'' lies in $\QACZ$, and then combine this with a ``weak-copy'' test, which estimates the hamming weight of the original string these weak copies were made from.

The ``poor-mans fanout'' circuit is an upgraded version of a controlled-$W$ map, which we use in its implementation.

\begin{claim}[Uniform $0^n$ and $W_{n}$]\label{cl:zerow}
For any $n$, the state $\ket{W'_n} = \frac{1}{\sqrt{2}} \ket{0^n} + \frac{1}{\sqrt{2}} \ket{W_n}$ can be cleanly prepared in $\QACZ$ using $O(n)$ qubits. 
\end{claim}
\begin{proof}
Start with the state $\ket{\psi_0}_T := \ket{1/p}^{\tens n}$ for $p = 1/(n+1)$. Observe that, 
\begin{align}
    p_1 := \Pr_{x \sim \bern(1/{n+1}, n)} \blr{|x| = 1} &= n \cdot (p) \cdot (1-p)^n \\
    &= (1-p)^{n} \\
    &= \Pr_{x \sim \bern(1/{n+1}, n)} \blr{|x| = 0}.  
\end{align}
Then, applying the $\THR^n_1$, which is in $\QACZ$, onto a new ancilla $a$, we can prepare, 
\begin{align}
    \ket{\psi_1}_{T,a} &= \THR^n_1 \ket{\psi_0}_T \ket{0}_a \\
    &=  \sqrt{p_1} \ket{W'_n}_T \ket{1}_a + \sqrt{1-p_1} \ket{\bad}_T \ket{0}_a.
\end{align}
From known binomial bounds (see \cref{fact:binombounds}), $\lim_{n \to \infty} p_1 = e^{-1} \approx 0.367$. Hence $p_1 \geq 1/4$ and applying \cref{cl:hiampamp} produces $\ket{W'_n}$ cleanly.
\end{proof}

\begin{lemma}[Controlled W map]\label{lem:ctrlW}
    We can apply the following map cleanly in $\QACZ$:
    $$\ket{0}_x \ket{0^n} \mapsto \ket{0}_x \ket{0^n}, \quad \ket{1}_x \ket{0^n} \mapsto \ket{1}_x \ket{W_n}.$$
\end{lemma}
\begin{proof}
One can construct the following state using an $\OR$ gate, then an additional $XZ$ gate on qubit $x$ from \cref{cl:zerow}. 
\begin{align}
   \ket{\varphi} = \frac{1}{\sqrt{2}} \ket{1}_x \ket{0^n}  - \frac{1}{\sqrt{2}} \ket{0}_x \ket{W_n}. 
\end{align}
We can implement $\R_{\varphi} = (I - 2\kb{\varphi})$, the reflection about this state in $\QACZ$ using \cref{fact:reflstate}. This has the effect of swapping the two branches.   
Then, on input $b \in \bin$, first construct the state, using \cref{cl:zerow}.
\begin{align}
    \ket{\psi_0(b)} &:= \ket{b}_x \ket{W'_n} \\
    &= \frac{1}{\sqrt{2}} \ket{b}_x \ket{0^n}  + \frac{1}{\sqrt{2}} \ket{b}_x \ket{W_n}  
\end{align}
Next, apply $\R_{\varphi}$ to obtain $\ket{\psi_1(b)}$. Observe that when $b = 1$, we have,
\begin{align}
     \ket{\psi_1(1)} &:= \R_{\varphi} \ket{\psi_0(1)} \\
     &= \frac{1}{\sqrt{2}} \lr{\R_{\varphi} \ket{1}_x \ket{0^n}}  + \frac{1}{\sqrt{2}} \ket{1}_x \ket{W_n}  \\
     &= \frac{1}{\sqrt{2}}  \ket{0}_x \ket{W_n}  + \frac{1}{\sqrt{2}} \ket{1}_x \ket{W_n}   \\
     &= \ket{+}_x \ket{W_n}
\end{align}
Because the $\ket{1} \ket{W_n}$ branch is orthogonal to $\ket{\varphi}$. 
Similarly, when $b = 0$, we obtain the state $\ket{\psi_1(0)} = \ket{+}_x \ket{0^n}$. To make this map clean, we can apply an $H$ gate on $x$ and then another $\OR$ gate to recover the input value of $b$. 
\end{proof}

\begin{corollary}[Any superposition of $\ket{0^n}$, $\ket{W_n}$] \label{cor:any0W}
For any $\alpha, \beta$ such that $\vlr{\alpha}^2 + \vlr{\beta}^2 = 1$, we can prepare the state, 
$\ket{W_n(\alpha, \beta)} := \alpha \ket{0^n}  + \beta \ket{W_n}$ in $\QACZ$. 
\end{corollary}
\begin{proof}
To prepare the desired superposition, start with a qubit $\ket{\nu}_x = \alpha \ket{0} + \beta \ket{1}$ and apply the circuit from \cref{lem:ctrlW}. To uncompute qubit $x$ apply an $\OR$ gate. 
\end{proof}

\subsection{Poor man's fanout}\label{sec:majupper}

\begin{lemma}[Uncompute $\ket{W_n}$ but not $\ket{0^n}$]\label{lem:uncomputeW}
There is a $\QACZ$ circuit $C$ that cleanly achieves the mapping, 
$$C \ket{1/n}^{\tens n} \ket{1}  = \ket{W^n} \ket{1} \qquad \text{and} \qquad C \ket{0^n} \ket{0}  = \ket{0^n} \ket{0}$$
\end{lemma}
\begin{proof}
First we will describe a circuit $C_0$ to go from $\ket{1/n}^{\tens n} \ket{1}$ to $\ket{W} \ket{1}$. Then we will argue that $C_0$ acts as desired on $\ket{0} \ket{0}$. First using a fresh ancilla, $a_0$, obtain, 
\begin{align}
    \ket{\psi_1} &:= \EXACT_1(X,a_0) \ket{1/n}^{\tens n}_X \ket{1}_q \ket{0}_{a_0} \\
    &= \sqrt{p} \ket{W_n} \ket{1}_q \ket{1}_{a_0} + \sqrt{1-p} \ket{\bad} \ket{1}_{q} \ket{0}_{a_0}.
\end{align}
where $p = \Pr[\binomd(1/n,n) = 1] \geq e^{-1}$. Then, similar to \cref{cl:hiampamp}, apply a controlled-$\Rot_{\gamma}$ gate controlled on both $q$ and $a_0$ to obtain, 
\begin{align}
    \ket{\psi_2} &:= \text{cc-}\Rot_{\gamma}(q,a_0,a_1) \ket{\psi_1} \ket{0}_{a_1} \\
    &= \frac{1}{2} \ket{W_n} \ket{1}_q \ket{1}_{a_0} \ket{1}_{a_1} + \frac{\sqrt{3}}{2} \ket{\bad}_{T,q,a_0} \ket{0}_{a_1}.  
\end{align}
Let $\R_{\ket{\psi_2}} := (I - 2\kb{\psi_2})$ be the reflection about $\ket{\psi_2}$. 
Observe that since $\ket{\psi_2}$ can be constructed in $O(1)$ depth,  $\R_{\ket{\psi_2}}$ is also in $\QACZ$. 
Next, obtain, 
\begin{align}
    \ket{\psi_3} &:=  Z_{a_0} \R_{\ket{\psi_2}} Z_{a_0} \ket{\psi_2} \\
    &= - Z_{a_0} \R_{\ket{\psi_2}} \cdot \lr{\frac{1}{2} \ket{W_n} \ket{1}_q \ket{1}_{a_0} \ket{1}_{a_1} - \frac{\sqrt{3}}{2} \ket{\bad}_{T,q,a_0} \ket{0}_{a_1}  } \\
    &= - Z_{a_0} \ket{W_n} \ket{1}_q \ket{1}_{a_0} \ket{1}_{a_1}  \\
    &= \ket{W_n} \ket{1}_q \ket{1}_{a_0} \ket{1}_{a_1}.  
\end{align}
On input registers marked $X,q$ and ancillae $\ket{00}_{a_0,a_1}$, this circuit is given by, 
\begin{align}
   C_0 := Z_{a_0} \R_{\ket{\psi_2}} Z_{a_0} \cdot \text{cc-}\Rot_{\gamma} \cdot \EXACT_1(X,a_0).
\end{align}
Then, observe that all these gates are ineffective on input $\ket{0^n}_X \ket{0}_{q}$. In particular, $\ket{\psi_2}$ is orthogonal to $\ket{0^n}_X \ket{0^3}_{q,a_0,a_1}$. Therefore, 
\begin{align}
&    C_0 \cdot \ket{0^n}_X \ket{0}_q \ket{00}_{a_0,a_1} =  \ket{0^n}_X \ket{0}_q \ket{00}_{a_0,a_1} \\
&    C_0 \cdot \ket{1/n}^{\tens n}_X \ket{1}_q \ket{00}_{a_0,a_1} = \ket{W_n}_X \ket{1}_q \ket{11}_{a_0,a_1}
\end{align}
Thus, $C := \cnot(q,a_0) \cnot(q,a_1) C_0$ implements the desired mapping.
\end{proof}

Then, the reverse circuit for the above is our ``poor man's'' $n$-fanout.

\poormansfanoutmap
\begin{proof}
Let $C$ be the circuit from \cref{lem:uncomputeW}. Then, this mapping is achieved by $C^\dag$.
\end{proof}

\subsection{Majority}
For ease of description we will use the following claim formally connecting classical post-processing with limited fanout to $\QACZ$ circuits.

\begin{claim}[Closure under $\ACZ$ postprocessing]\label{cl:ac0postprocessing}
Let $U$ be a $\QACZ$ circuit that, on input $\x$, produces output qubits whose computational-basis measurement outcomes form a distribution $\D(\x)$ over $\bin^{m(n)}$. Suppose there is an $\ACZ$ circuit $T$ with $\polylog(n)$ fanout satisfying, for all $\x \in \bin^n$, 
$$\Pr_{\y \sim \D(\x)}[T(\y)) = g(\x)] \ge 1-\delta$$
for a function $g(\x)$, then there is a $\poly(m)$ size single-output $\QACZ$ circuit $C$ (\cref{def:booleanckt}) such that, $$\corr(f_C, g) \geq 1-2\delta.$$
\end{claim}
\begin{proof}
Using \cref{fact:polylogfan}, there is a $\QACZ$ circuit $C'$ acting on $m$ inputs that, using $\poly(m)$ ancillae, implements the following map for all $\y \in \bin^m$, 
\begin{align}
   \ket{\y}_X \ket{0}_t \mapsto \ket{\y}_X \ket{T(\y)}_t
\end{align}
Now, observe that any measurement on the $X$ register in the output does not affect measurements on the $t$ register (since they commute).  Prepend $U$ to $C'$, so that the $m$ output qubits of $U$, feed into the inputs of $C'$ to obtain $C$ that acts on $n$ inputs. 

Then, on any input $\x \in \bin^n$, apply $C$ and measure both $t$ and $X$ in the standard basis. The distribution on $t$ is exactly the one obtained by applying $T(\y)$ on a sample from $\D(\x)$. 
Thus the resulting circuit outputs $g(\x)$ with probability at least $1-\delta(n)$ on every input, which corresponds to correlation $1-2\delta(n)$ with $g(\x)$.
\end{proof}

Poor man's fanout allows us to create $n$ ``weaker'' copies of the input. In order to approximate majority, we can separately apply $\QACZ$ circuits to each copy to estimate majority and amplify. A gadget with similar properties to poor man's fanout was given in Appendix B of \cite{grier2026tc0}. The lemma below, which will serve as the key ingredient in our test, is analogous to their Claim 28.

\oneweakcopytest
%
\begin{proof}
Let $\ket{\phi} := a \ket{0^n}+ b\ket{W_n}$ for some $a,b$ with $a^2+b^2=1$ we will choose later, and let $\R = (I - 2\kb{\phi})$, the reflection about $\ket{\phi}$. Then, the circuit $C$ first applies $\R$ and then a $\OR$ gate with $o$ as the target. Mathematically, $C$ is given by, 
\begin{align}\label{eq:cktdef}
    C := (I - 2\kb{0^n} \tens \kb{+}_o) \cdot \R
\end{align}
From \cref{cor:any0W} one can implement any such $\ket{\phi}$ therefore implement $\R$ using \cref{fact:reflstate} in $\QACZ$. 

We will choose $a,b$ as follows. First observe for any $\x$ 
of hamming weight $\ell$, $\alpha_{\ell} := \braket{0^n | \psi(\x)}$ and $\gamma_{\ell} := \braket{W_n | \psi(\x)}$ depend only on $\ell$. 
Let $s :=\sqrt{\alpha_t^2+\gamma_t^2}$, observe $s^2$ is the total weight on the subspace spanned by $\ket{0^n}, \ket{W_n}$ in $\ket{\psi(\x)}$ whenever $|\x|=t$. Choose, 
\begin{align}
a=\sqrt{\frac{s-\gamma_t}{2s}}, \qquad b=\sqrt{\frac{s+\gamma_t}{2s}}
\end{align}
so that $1-2a^2=\gamma_t/s$ and $2ab=\alpha_t/s$. 
This defines the circuit $C$ through \cref{eq:cktdef}.

Now consider any input $\x \in \bin^n$  and let $\ell = |\x|$. Then, 
$\ket{\psi(\x)}$ can be written as, 
\begin{align}\label{eq:inpdec2}
\ket{\psi_\ell}=\alpha_\ell\ket{0^n}+\gamma_\ell\ket{W_n}+\ket{\ov{\eta(\x)}}
\end{align}
for some un-normalized $\ket{\ov{\eta(\x)}}$ with $\braket{0^n | \ov{\eta(\x)}} = \braket{W_n | \ov{\eta(\x)}}$. 
Since the final $\OR$ gate flips the ancilla $o$ only in the subspace spanned by $\ket{0^n}$, the probability of measuring $\ket{1}$ on the output is, $\braket{1 | \rho(\x) | 1} = \vlr{\braket{0^n | \R | \psi(\x)}}^2$, where,
\begin{align}
 \braket{0^n | \R | \psi(\x)} &= \bra{0^n} \cdot (I - 2\kb{\phi}) \ket{\psi(\x)} \\
    &=  2 \bra{0^n} \cdot \kb{\phi} \cdot \lr{\gamma_{\ell} \ket{W_n} + \ket{\ov{\eta}(\x)}}  + \alpha_{\ell} \braket{0^n | \R | 0^n} \\
    &=  2 \gamma_{\ell} \bra{0^n} \cdot \kb{\phi} \cdot  \ket{W_n}  + \alpha_{\ell} \braket{0^n | \R | 0^n} \\
    &= 2 \gamma_{\ell} ab + \alpha_\ell \cdot (1-2a^2) \\
    &= \gamma_{\ell} \cdot \alpha_t/s + \alpha_{\ell} \cdot \gamma_t/s
\end{align}
Now, letting $r:=\sqrt{1-\frac1n}$, gives $\alpha_\ell = r^{\ell}$. 
Observing there are exactly $\ell$ hamming weight $1$ basis vectors with non-zero amplitudes of $\frac1{\sqrt n}r^{\ell-1}$  in $\ket{\psi(\x)}$ gives $\gamma_\ell = \frac{\ell}{n}r^{\ell-1}$. 
Substituting, 
\begin{align}
s &=\sqrt{r^{2t}+\frac{t^2}{n^2}r^{2t-2}} \\ 
    &=r^{t-1}\sqrt{r^2+\frac{t^2}{n^2}} 
\end{align}
gives, 
\begin{align}
 \braket{0^n | \R | \psi(\x)} = \frac{t-\ell}{n}\cdot \frac{r^\ell}{\sqrt{r^2+\frac{t^2}{n^2}}}.
\end{align}
Now assume $|\ell-t|=O(\sqrt n\,\polylog n)$. Then $(\ell-t)/n=o(1)$, hence
$r^\ell=r^t(1+o(1))$. Therefore
\begin{align}
\braket{1 | \rho(\x) | 1} = (1+o(1))\left(\frac{\ell-t}{n}\right)^2 \frac{r^{2t}}{r^2+\frac{t^2}{n^2}}.
\end{align}
Since $r^{2t}=(1-\frac1n)^t=e^{-t/n}(1+o(1))$ and
$r^2+\frac{t^2}{n^2}=1+\frac{t^2}{n^2}+o(1)$, it follows that
\begin{align}
\braket{1 | \rho(\x) | 1} = \lambda_t(1+o(1))\left(\frac{\ell-t}{n}\right)^2,
\qquad
\lambda_t:=\frac{e^{-t/n}}{1+t^2/n^2}=\Theta(1).
\end{align}
Therefore, as claimed,
\begin{align}
\braket{0 | \rho(\x) | 0} &= 1-\braket{1 | \rho(\x) | 1} = 1-\lambda_t(1+o(1))\left(\frac{\ell-t}{n}\right)^2.
\end{align}
\end{proof}

\begin{lemma}[$\QACZ$ circuit for $\APPROX_t$]\label{lem:approxtcircuit}
Let $a=\polylog(n)$ and let $t\in\{0,\dots,n\}$. There exists a $\QACZ$ circuit $C$ such that, on every input $\x\in\{0,1\}^n$ of Hamming weight $\ell$, $C$ outputs $1$ with probability at least $1-1/n$ if $|\ell-t|<\frac{\sqrt n}{2a}$, and outputs $0$ with probability at least $1-1/n$ if $|\ell-t|>\frac{\sqrt n}{a}$.
\end{lemma}
\begin{proof}
Let $C_0$ be the circuit from \Cref{lem:onecopytest}, and let $Z$ be the bit that is $1$ iff the last qubit of $C_0$ measures to $\ket{1}$. By \Cref{lem:onecopytest}, $\Pr[Z=1]=\lambda_t(1+o(1))((\ell-t)/n)^2$, where $\lambda_t=\frac{e^{-t/n}}{1+t^2/n^2}$. Thus, for all large enough $n$, we may bound the $1+o(1)$ factor between $2/3$ and $4/3$, so
\[
\frac{2}{3}\lambda_t\left(\frac{\ell-t}{n}\right)^2
\le
\Pr[Z=1]
\le
\frac{4}{3}\lambda_t\left(\frac{\ell-t}{n}\right)^2.
\]
Hence $\Pr[Z=1]\le \lambda_t/(3na^2)$ when $|\ell-t|<\frac{\sqrt n}{2a}$, and $\Pr[Z=1]\ge 2\lambda_t/(3na^2)$ when $|\ell-t|>\frac{\sqrt n}{a}$.

Set $r=\lceil 3na^2/\lambda_t\rceil$, $m=d\log n$ for a large enough constant $d$, and $s=\lceil rm/n\rceil=\polylog(n)$. Using $\polylog(n)$ fanout, make $s$ copies of the input string $\x$. For each copied string and each $i\in[n]$, apply \Cref{lem:poormanfanout} to the copied bit $x_i$. For each $u\in[n]$, the $u$th fanout outputs across $i=1,\dots,n$ form one copy of the encoded input $\ket{x_1/n}\cdots\ket{x_n/n}$. Thus we obtain $sn\ge rm$ unentangled encoded inputs, and hence $rm$ independent copies $Z_{j,k}$ of $Z$, with $j\in[m]$ and $k\in[r]$.

Let $W_j=\bigvee_{k=1}^r Z_{j,k}$. In the first case, $\Pr[W_j=1]\le r\Pr[Z=1]\le 2/3$. In the second,
\[
\Pr[W_j=1]=1-(1-\Pr[Z=1])^r \ge 1-\left(1-\frac{2\lambda_t}{3na^2}\right)^r \ge 1-e^{-2}.
\]
Since the $W_j$ are independent, a Chernoff bound gives $\sum_j W_j<0.7m$ in the first case and $\sum_j W_j>0.8m$ in the second, each with probability at least $1-1/n$, for large enough $d$.

By a well-known result of Ajtai and Ben-Or~\cite{ajtai1984probabilistic}, there is an approximate-majority function on $m$ bits, that is, a function $A_m:\{0,1\}^m\to\{0,1\}$ computable by an $\ACZ$ circuit of size $\poly(m)$, which outputs $1$ whenever the input Hamming weight is below $0.7m$ and outputs $0$ whenever it is above $0.8m$. Since $m=O(\log n)$, this circuit has size, and hence fanout, at most $\polylog(n)$. Applying $A_m$ to $W_1,\dots,W_m$ therefore gives the desired output with probability at least $1-1/n$, and the result follows from \cref{cl:ac0postprocessing}.
\end{proof}

\ezmajority

\begin{proof}
Let $\gamma=\frac{\sqrt n}{a}$ and $\eta=\frac{\sqrt n}{2a}$. For each threshold $t$, let $D_t$ be the $\QACZ$ circuit from \cref{lem:approxtcircuit} . Thus, on inputs of Hamming weight $\ell$, $D_t$ outputs $1$ with probability at least $1-1/n$ if $|\ell-t|<\eta$, and outputs $0$ with probability at least $1-1/n$ if $|\ell-t|>\gamma$.

Let $L=c'\log{n}\sqrt{n}$ for a sufficiently large constant $c'$, and consider thresholds
\[
T_+=\{n/2+j\eta: 1\le j\le L/\eta\},
\qquad
T_-=\{n/2-j\eta: 1\le j\le L/\eta\}.
\]
Since $|T_+\cup T_-|=\polylog(n)$, we may run all the circuits $D_t$ in parallel. Let $C$ output $1$ if some $D_t$ with $t\in T_+$ outputs $1$, output $0$ if some $D_t$ with $t\in T_-$ outputs $1$, and output an arbitrary bit otherwise. By \Cref{cl:ac0postprocessing}, this is computable in $\QACZ$.

If $\gamma<\ell-n/2\le L$, then some $t\in T_+$ satisfies $|\ell-t|<\eta$, while every $t\in T_-$ satisfies $|\ell-t|>\gamma$. Hence, except with probability at most $|T_+\cup T_-|/n$, the circuit $C$ outputs $1=\MAJ_n(x)$. The case $-L\le \ell-n/2<-\gamma$ is symmetric.

Thus $C$ can disagree with $\MAJ_n$ only if one of the above tests outputs the wrong answer, if $|\ell-n/2|\le \gamma$, or if $|\ell-n/2|>L$. The first case contributes at most $|T_+\cup T_-|/n=\polylog(n)/n$ failure probability. To analyze the second case, for a uniformly random input, the probability that the hamming weight $\ell$ satisfies $|\ell-n/2|\le \gamma$ is at most $O(\gamma/\sqrt n)=O(1/a)$. For the last case, the probability that $\ell$ satisfies $|\ell-n/2|>L$ is at most $2e^{-2L^2/n}=O(1/n)$. Each of these contributions is $O(1/a)$. Therefore $C$ agrees with $\MAJ_n$ on a $1-O(1/a)$ fraction of inputs, and hence has correlation at least $1-1/a$ after adjusting constants.
\end{proof}

\ifnum\ANON=1
\else
\section{Acknowledgments}
The authors are grateful to Avishay Tal and Ryan Williams for valuable advice and discussions regarding this work. 
\fi

\appendix
\bibliographystyle{alpha}
\bibliography{main}

\section{Skipped proofs}\label[appendix]{sec:addproofs}
\approxcleancomp*
\begin{proof} [Proof of \cref{fact:cleancomp}]
Preserve the input registers by making a classical copy of each input coordinate $x_i$ in a single layer of $\cnot$ gates at the start. This increases the number of ancillae to $a = n+m$. Call this circuit $C_0$. Let $A$ be the set of ancilla qubits of $C_0$ other than output register $t$, then, wlog, the final state of $C_0$ has the form,
  \begin{align}
    C_0 \ket{\x} \ket{0^a}_{t,A} &= \ket{\x} \lr{\sqrt{p_0(\x)}  \ket{0}_t \ket{\psi_0(\x)}_A + \sqrt{p_1(\x)} \ket{1}_t \ket{\psi_1(\x)}_A} 
  \end{align}
for some arbitrary input-dependent states $\ket{\psi_0(\x)}$ and $\ket{\psi_1(\x)}$ on the remaining ancillae that contain global phases to absorb the additional phases. Now using another fresh ancilla $t'$, define $C'$ using the standard copy-uncompute procedure as,  
\begin{align}
    C' := C_0^\dag \cdot \cnot(t,t') \cdot C_0
\end{align}
Then, 
\begin{align}
\cnot(t,t') C_0 \ket{\x} \ket{0^a}_{A} \ket{0}_{t'} &= \ket{\x} \underbrace{\lr{ \sqrt{p_0(\x)} \ket{0}_t \ket{\psi_0(\x)}_A \ket{0}_{t'} + \sqrt{p_1(\x)} \ket{1}_t \ket{\psi_1(\x)}_A \ket{1}_{t'}}}_{\ket{\varphi(\x)}}.  
\end{align}
Therefore, for both values of $\b \in \bin$, the claimed amplitude is,
\begin{align}
    \alpha_b(\x) &:= \bra{\x} \bra{0^a} \bra{b}_{t'} \cdot C' \cdot \ket{\x} \ket{0^{a+1}}_{A,t'} \\
&= \lr{\bra{\x} \bra{0^a} \bra{b}_{t'} \cdot C_0^\dag} \cdot \ket{\x} \ket{\varphi(\x)}_{A,t'} \\
&= \lr{\bra{\x} \bra{0^a} \bra{b}_{t'} \cdot C_0^\dag} \cdot \lr{\sqrt{p_b(\x)} \cdot \ket{\x} \ket{b}_t \ket{b}_{t'} \ket{\psi_b(\x)}_A}\\ 
&=  \lr{\sqrt{p_0(\x)} \bra{0}_t \bra{\psi_0(\x)}_A + \sqrt{p_1(\x)} \bra{1}_t \bra{\psi_1(\x)}}_A \cdot \sqrt{p_b(\x)} \ket{b}_t \ket{\psi_b(\x)}_A \\
&= p_b(\x).
  \end{align}
\end{proof}

\begin{lemma}[Felinity of uniformly rotated W states]\label{lem:felwstate}
For any state $\ket{\psi}$ of the form $\ket{\psi} := U^{\tens n} \ket{W_n}$ for any unitary $U$, 
$$\feln(\psi) = \negl(n).$$
\end{lemma}
\begin{proof}
$U$ can be written in its Euler angle decomposition as, 
\begin{align}
    U = e^{i\delta} R_z(\alpha) R_y(\beta) R_z(\gamma). 
\end{align}
Observe that any global phases and any matrices that commute with $Z$ do not affect the felinity, since it is defined with respect to $Z$-basis measurements. Therefore, all the $R_z$ rotations commute and it suffices to analyze $R_y(\beta)$, i.e,  the state  $\ket{\varphi} := R_y(\beta)^{\tens n} \ket{W_n}$.

$\ket{\varphi}$ lies in the symmetric subspace, and we can express it in the Dicke basis:
\begin{align}
\ket{\varphi} = \sum_{k=0}^n \alpha_k \ket{D_k^n},
\end{align}
Then, 
\begin{align} 
\feln(\varphi) &= 2 \cdot \sum_{k=0}^n \binom{n}{k}^{-1} \cdot \vlr{\alpha_k}^2 \cdot \vlr{\alpha_{n-k}}^2 \label{eq:here1324}
\end{align}
Identify the symmetric subspace with the spin-$j=n/2$ irreducible representation of $SU(2)$, under which
$\ket{D_k^n}$ corresponds to the angular momentum eigenstate $\ket{j=n/2,\,m=k-n/2}$. Thus,
\begin{align}
\alpha_k &= \bra{D_k^n}R_y(\beta)^{\otimes n}\ket{D_1^n} \\
&= d^{\,n/2}_{\,k-n/2,\;1-n/2}(\beta),
\end{align}
where $d^j_{m',m}(\beta)$ denotes the Wigner small-$d$ matrix. We require the entry for $m' = k-n/2$ and $m = 1-n/2$. These matrix elements are well known and admit explicit closed-form expression. For e.g., see section (4.3.4) of \cite{varshalovich1988quantum} which provides the entry we need in terms of Jacobi polynomials. From this, one obtains,  
\begin{align}
\vlr{\alpha_k} \leq  \poly(n) \cdot \sqrt{\binom{n}{k}} \cdot \cos^{n-k-1} \lr{\frac{\beta}{2}}  \sin^{k-1}\lr{\frac{\beta}{2}}
\end{align}
Then, since $\sin^2 \theta \cos^2 \theta$ is at most $1/4$, each term of the summation in \cref{eq:here1324} can be bounded by, 
\begin{align}
\binom{n}{k}^{-1}|\alpha_k|^2 |\alpha_{n-k}|^2  &\le \poly(n) \lr{\sin(\beta/2) \cdot \cos(\beta/2)}^{2n-4}  \\
&\leq \poly(n) (1/2)^{2n-4} \\
&= \negl(n)
\end{align}
Therefore,
\begin{align}
    \feln(\psi) = \feln(\varphi) = \negl(n) 
\end{align}
\end{proof}

\end{document}